\documentclass[sigconf,screen]{acmart} 

\AtBeginDocument{%
  \providecommand\BibTeX{{%
    \normalfont B\kern-0.5em{\scshape i\kern-0.25em b}\kern-0.8em\TeX}}}

\setcopyright{acmcopyright}
\copyrightyear{2018}
\acmYear{2018}
\acmDOI{10.1145/1122445.1122456}

\acmConference[Woodstock '18]{Woodstock '18: ACM Symposium on Neural
  Gaze Detection}{June 03--05, 2018}{Woodstock, NY}
\acmBooktitle{Woodstock '18: ACM Symposium on Neural Gaze Detection,
  June 03--05, 2018, Woodstock, NY}
\acmPrice{15.00}
\acmISBN{978-1-4503-XXXX-X/18/06}

\renewcommand\footnotetextcopyrightpermission[1]{} 

\usepackage[utf8]{inputenc}
\usepackage{xspace}
\usepackage{amsmath}
\usepackage{amsthm}
\usepackage[algo2e,ruled,vlined,linesnumbered]{algorithm2e}
\usepackage{url}

\usepackage{todonotes}

\newtheorem*{example*}{Example}

\newcommand{\comp}{\textsc{comp}\xspace}

\newcommand{\neighs}{\textsc{neighbors}\xspace}

\newcommand{\nph}{\textsc{np-h}ard\xspace}

\newcommand{\ccomp}{\texttt{cc}\xspace}

\newcommand{\scap}{\ensuremath{\tilde{\cap}}\xspace}

\newcommand{\ocost}{\mathcal{O}_T}
\newcommand{\ccost}{\mathcal{C}_t}

\newcommand{\sol}{\ensuremath{\mathcal{S}}\xspace}
\newcommand{\rec}{\textsc{enum}\xspace}
\newcommand{\dv}{\ensuremath{\dot v}\xspace}
\newcommand{\de}{\ensuremath{\dot e}\xspace}

\newcommand{\nsol}{\ensuremath{\mathcal{N}}} 

\newcommand{\iack}{\ensuremath{\alpha}\xspace} 

\newcommand{\restrtime}{\mathcal{R}_T}
\newcommand{\restrbound}{\mathcal{R}_N}

\begin{document}

\title{Proximity Search For Maximal Subgraph Enumeration}

\author{Alessio Conte}
\affiliation{%
  \institution{University of Pisa, Italy}
  \streetaddress{}
  \city{} 
  \country{}
  \postcode{}
}
\email{conte@di.unipi.it}

\author{Roberto Grossi}
\affiliation{%
  \institution{University of Pisa, Italy}
}
\email{grossi@di.unipi.it}

\author{Andrea Marino}
\affiliation{%
  \institution{University of Florence, Italy}
}
\email{andrea.marino@unifi.it}

\author{Takeaki Uno}
\affiliation{%
  \institution{National Institute of Informatics, Japan}
  \streetaddress{}
  \city{} 
  \country{}
  \postcode{}
}
\email{uno@nii.ac.jp}

\author{Luca Versari}
\affiliation{%
  \institution{University of Pisa, Italy and Google Research}
}
\email{veluca@google.com}

\begin{abstract}
This paper proposes a new general technique for maximal subgraph enumeration which we call \textit{proximity search}, whose aim is to design efficient enumeration algorithms for problems that could not be solved by existing frameworks.
To support this claim and illustrate the technique we include output-polynomial algorithms for several problems for which output-polynomial algorithms were not known, including the enumeration of Maximal Bipartite Subgraphs, Maximal $k$-Degenerate Subgraphs (for bounded $k$), Maximal Induced Chordal Subgraphs, and Maximal Induced Trees.
Using known techniques, such as reverse search, the space of all maximal solutions induces an implicit directed graph called ``solution graph'' or ``supergraph'', and solutions are enumerated by traversing it; however, nodes in this graph can have exponential out-degree, thus requiring exponential time to be spent on each solution.
The novelty of proximity search is a formalization that allows us to define a better solution graph, and a technique, which we call \textit{canonical reconstruction}, by which we can exploit the properties of given problems to build such graphs.
This results in solution graphs whose nodes have significantly smaller (i.e., polynomial) out-degree with respect to existing approaches, but that remain strongly connected, so that all solutions can be enumerated in polynomial delay by a traversal.
A drawback of this approach is the space required to keep track of visited solutions, which can be exponential:
we further propose a technique to induce a parent-child relationship among solutions and achieve polynomial space when suitable conditions are met.
\footnote{A preliminary version of this paper has been presented in~\cite{Conte2019proximity}, containing some of the exponential-space algorithms. In this extended version, we factorize the fundamental principles of algorithms in~\cite{Conte2019proximity} to define a general guide for designing a proximity search algorithm, which we call \textit{canonical reconstruction}. We also introduce a general technique to obtain proximity search algorithms that use polynomial space, as well as proximity search algorithms for new problems, both with exponential-space and polynomial-space requirements. A draft of this extended version is also available at~\cite{conte2019proximityArxiv}.}

\end{abstract}



\maketitle

\section{Introduction}

Given a universe of elements, such as the vertices or edges of a graph, and a property on them, such as being a clique or a tree, a listing problem asks to return all subsets of the universe which satisfy the given property. 

Listing structures, within graphs or other types of data, is a basic problem in computer science, and it is at the core of data analysis. While many problems can be solved by optimization approaches for the best solution, e.g., by finding the shortest path, or the largest clique, others require finding several solutions to the input problem: in community detection, for example, finding just one ``best'' community only gives us local information regarding some part of the data, so we may want to find several communities to make sense of the input.
Furthermore, many real-world scenarios may not have a clear objective function for the best solution: We may define an algorithm to optimize some desired property, but the optimal solution found may be lacking further properties that emerge during listing or simply not be practical. 
We may want instead to quickly list several solutions, suitable according to some metrics, then analyze them a posteriori to find the desired one.

In these scenarios, listing only the solutions that are \textit{maximal} under inclusion is a common-sense requirement whenever it can be applied,\footnote{In other problems, we may want \textit{minimal} solutions instead, although this is usually an equivalent concept, as it corresponds to the complement of a solution being maximal.} as maximal solutions subsume the information contained in all others, and can be exponentially fewer: For example, a graph may have up to $2^n$ cliques, but only $3^{n/3}$ maximal ones~\cite{moon1965cliques}. For brevity, we call \textit{maximal listing problem} a listing problem where only the inclusion-maximal solutions should be output.

From a theoretical point of view, listing provides many challenging problems, especially when maximality is required. When dealing with listing algorithms, we are often interested in their complexity with respect to both $n$, the input size, and $\nsol$, the size of the output. Algorithms whose complexity can be bounded by a polynomial of these two factors are called \textit{output polynomial} or \textit{polynomial total time}~\cite{JOHNSON1988119}.
Interestingly, the hardness of listing problems does not seem to be correlated with that of optimization: there are several \nph maximum optimization problems whose corresponding maximal listing problem admits an output-polynomial solution (see, e.g.,~\cite{tsukiyama1977new,avis1996reverse}); on the other hand, there are problems for which one maximal (or maximum) solution can be identified in polynomial time, but an output-polynomial algorithm for listing maximal solutions would imply \textsc{p=np}~\cite{lawler1980generating}.

A long-standing question in the area is to find a characterization of which listing problems allow for output-polynomial solutions and which do not. 
Furthermore, within output-polynomial algorithms stricter complexity classes exist, such as \textit{incremental polynomial time}, where the time to output the $i$-th solution is polynomial in $n$ and $i$, and \textit{polynomial delay}, where the time elapsed for outputting the next solution is upper bounded by a polynomial in $n$.
The latter class is of particular interest in practical scenarios, as it guarantees that solutions are output at a regular pace.

In this paper we add a few points to the latter class, by showing that there exist polynomial delay algorithms for some subgraph listing problems. More formally, we prove Theorem~\ref{thm:main}.


\begin{theorem}\label{thm:main}
The following problems allow polynomial delay listing algorithms by proximity search:


\begin{center}
\begin{small}
\begin{tabular}{cc}
\textsc{problem} & \textsc{delay}\\
maximal induced bipartite sg. & $O(n( m + n\iack(n)))$ \\
maximal connected induced bipartite sg. & $O(mn)$\\
maximal bipartite edge-induced sg. & $O(m^3)$ \smallskip\\
maximal induced $k$-degenerate sg. & $O(mn^{k+2})$\\
maximal edge-induced $k$-degenerate sg. & $O(m^3n^{k-1})$ \smallskip\\
maximal induced chordal sg. & $O(m^2n)$ \\
maximal connected induced chordal sg. & $O(m^2n)$ \\
maximal edge-induced chordal sg. & $O(m^4n)$ \smallskip\\
%
%
%
%
%
%
maximal induced proper interval sg. & $O(m^2n^3)$ \\
maximal connected induced proper interval sg. & $O(mn^3)$\smallskip\\
%
%
maximal connected obstacle-free convex hulls  & $O(n^4)$ \smallskip\\
maximal induced trees & $O(m^2)$ \smallskip\\
maximal connected induced directed acyclic sg. & $O(mn^2)$ \\
maximal connected edge-induced directed acyclic sg. & $O(m^3)$\\
    \end{tabular}

\noindent
Where ``sg.'' stands for subgraphs, $n$ and $m$ are the number of vertices and edges, $\iack(\cdot)$ is the functional inverse of the Ackermann function~\cite{Tarjan75UF}. 

All the algorithms use $O(\nsol n)$ space, where $\nsol$ is the number of solutions.\\
\end{small}
\end{center}
\end{theorem}

To the best of our knowledge, no output-polynomial result was previously known for these problems. For completeness, we consider both \textit{induced} subgraphs (i.e., sets of vertices) and \textit{edge-induced} subgraphs (i.e., sets of edges), as well as the \textit{connected} case where solutions are required to be connected, as the structure of such variants can differ significantly.

Furthermore, we abstract a general technique that can be used to obtain similar results on other problems. 
We do so by defining a graph whose vertices are the maximal solutions to the listing problem, and with directed edges between pairs of solutions, which we call \textit{solution graph}. The listing problem is solved by traversing the solution graph, and proving that all solutions are found this way.
The concept of solution graph is common to existing approaches, and general techniques already exist for building them, e.g.,~\cite{Cohen20081147}. However, the solution graph built with known approaches such as~\cite{Cohen20081147} may have too many edges, resulting in a traversal with exponential delay.

The key concept given in this paper is a technique to build a solution graph with fewer edges, while proving that all solutions are still found by its traversal.
An interesting property of this approach is that the resulting algorithms are remarkably simple to implement, while the complexity lies in proving their correctness.
We call this technique \textit{proximity search} since at its core lies a problem-specific notion of proximity. This notion acts as a sort of compass on the solution graph built by our algorithm, as given any two solutions $S$ and $S^*$, we will show that we always traverse an edge from $S$ to another solution $S'$ that has higher proximity to $S^*$; as $S^*$ has the highest proximity to itself, this implies that a traversal of the solution graph from any solution finds all others.
While others, such as~\cite{Cohen20081147,schwikowski2002enumerating}, already used the principle of reachability in the solution graph, we aim to define a looser set of necessary condition in order to guarantee this reachability, allowing more freedom in the design of algorithms, while at the same time formalize a technique called \textit{canonical reconstruction} that is effective in decomposing the structure of several problems to fit these rules.
The combination of these two parts creates algorithms that overcome the exponential burden imposed by the so-called \textit{input-restricted problem}, a reduced instance of the original problem that dominates the cost per solution of such approaches whose cost may be inherently exponential.

While the space required for a traversal of the graph is inherently proportional to the number of solutions, i.e., can be exponential in $n$, some output-polynomial techniques such as reverse search are able to work in polynomial space by inducing a tree-like structure on the solution graph, provided that the problem at hand is hereditary (i.e.~its property holds for the induced subgraphs) and the input-restricted problem is solvable efficiently.

By adding suitable constraints to the problems considered, we show a technique that combines proximity search with a recent generalization of reverse search to non-hereditary problems~\cite{conte2019framework}, obtaining algorithms with both polynomial-delay and polynomial space for some instances of proximity search. In particular, we prove that:

\begin{theorem}\label{thm:pspace}
The following problems allow polynomial delay listing and polynomial space algorithms by proximity search, with the following bounds:

\smallskip

\begin{small}
\begin{center}
\begin{tabular}{ccc}
\textsc{problem} & \textsc{delay} & \textsc{space} \\
maximal induced bipartite sg. & $O(n^2( m + n\iack(n)))$ & $O(m)$ \\
maximal connected induced bipartite sg. & $O(mn^2)$ & $O(m)$ \smallskip\\
%
maximal obstacle-free convex hulls  & $O(n^4)$  & $O(n)$ \smallskip\\
maximal induced trees & $O(m^2n^2)$  & $O(m)$\\
maximal induced forests & $O(m^2n^2)$  & $O(m)$ \smallskip\\
%

    \end{tabular}
\end{center}
\noindent
Where notation is as in Theorem~\ref{thm:main}.
\end{small}
\end{theorem}





\subsection{Related Work}


The listing problems considered in this paper model solutions as sets of elements (e.g., sets of vertices or edges of a graph), and consist in listing sets of elements with some required property, e.g., inducing a bipartite subgraph, or a tree.
We observe that the output is a family of sets, we can associate properties with the corresponding set systems: for example, a property is hereditary when each subset of a solution is a solution, which corresponds to the well-known independence systems~\cite{lawler1980generating}.


In this context, a simple yet powerful technique is recursively partitioning the search space into all solutions containing a certain element, and all that do not. 
This technique, usually called binary partition or simply backtracking, proves efficient when listing all solutions ~\cite{Ruskey03combinatorialgeneration}, and can be used to design algorithms that are fast in practice,\footnote{E.g., implementations of the Bron-Kerbosh~\cite{tomita2006worst} algorithm tend to be faster than those of output-polynomial algorithms~\cite{DBLP:conf/icalp/ConteGMV16} for listing maximal cliques.} or that can bound the number of solutions in the worst-case~\cite{Fomin:2008:CBV:1435375.1435384}.
%
%
On the other hand, this strategy rarely gives output-polynomial algorithms when dealing with maximal solutions, as we may spend time exploring a solution subspace that contains many solutions but no maximal one.


To obtain output-polynomial algorithms for maximal solutions, many algorithms rely on the following idea: given a maximal solution $S$, and some element $x\not\in S$, the hardness of listing solutions \textit{maximal within} $S\cup \{x\}$ is linked to the hardness of listing them in a general instance. One of the earliest mentions of the idea can be found in the seminal paper by Lawler et al.~\cite{lawler1980generating}, that generalizes ideas from Paull et al.~\cite{paull1959minimizing} and Tsukiyama et al.~\cite{tsukiyama1977new}, and has been formally defined as \textit{input-restricted problem} by Cohen et al.~\cite{Cohen20081147}.

The intuition is that the solutions obtained this way, using a maximal solution $S$ and an element not in $S$, can be used to generate new maximal solutions of the original problem. We can thus traverse an implicit directed graph, which we will call \textit{solution graph}, where the vertices are the maximal solutions and the out-neighbors are obtained by means of the input-restricted problem.

In particular,~\cite{lawler1980generating} showed how solving this problem could yield an output-polynomial and polynomial space listing algorithm for properties corresponding to \textit{independence systems}, assuming the input-restricted problem has a bounded number of solutions. ~\cite{Cohen20081147} showed that the strategy could be extended to the more challenging \textit{connected-hereditary} graph properties (i.e., where \textit{connected} subsets of solutions are solutions) using exponential space, and recently,~\cite{conte2019framework} showed that the same result can be obtained in polynomial space for \textit{commutable} set systems (which include connected-hereditary properties).

A clear limitation of this approach is that, in order to obtain polynomial-delay algorithms, the input-restricted problem needs to be solved in polynomial time. This is possible for some problems (e.g., cliques and independent sets), but impossible for others, simply because their input-restricted problems may have exponentially many solutions. Figure~\ref{fig:bip:restr} shows an example for maximal bipartite subgraphs.

\begin{figure*}
    \centering
    \includegraphics[width=.9\textwidth]{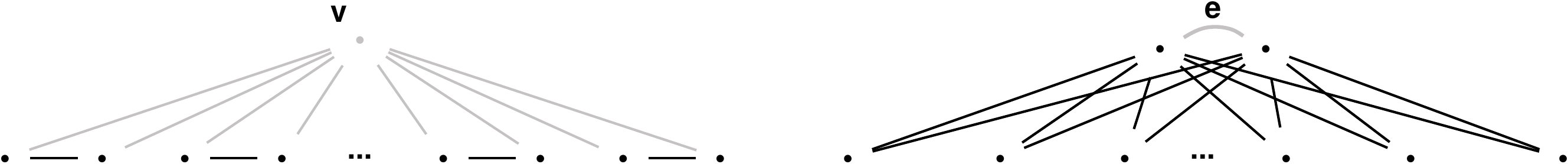}
    \caption{Instances of input-restricted problem for maximal bipartite subgraphs. 
    On the left: the black dots define a maximal bipartite induced subgraph; adding the vertex $v$ creates a graph with exponentially many maximal induced bipartite subgraphs, as we can obtain one by removing a vertex from each connected pair in the bottom in any combination. 
    On the right: the black edges define a maximal bipartite subgraph, and the addition of edge $e$ creates a graph with exponentially many maximal edge-induced bipartite subgraphs: every vertex on the bottom is incident to two edges; removing exactly one for each vertex yields a maximal edge-induced bipartite subgraph.}
    \label{fig:bip:restr}
\end{figure*}

The literature contains many more results concerning the enumeration of maximal/minimal solutions, e.g.,~\cite{avis1996reverse,koch1996algorithm,schwikowski2002enumerating,Golovach2018,GELY20091447,CARMELI2020}, and in particular regarding challenging problems such as the well-known minimal hypergraph transversals/dominating sets problem~\cite{kante2014enumeration,golovach2015incremental,elbassioni2009output}. 
However, to the best of our knowledge, the only two effective general techniques for listing maximal solutions in an output-sensitive fashion are the extension problem (binary partition, flashlight search), and the input-restricted problem: proximity search can be a valuable tool when the previous two fail.
We motivate this by showing the first polynomial delay algorithms for several maximal listing problems whose associated input-restricted problem is not solvable in polynomial time.

As mentioned above, a preliminary version of this paper containing some of the exponential-space algorithms has appeared in~\cite{Conte2019proximity}. Since its publication, some preprints~\cite{brosse2020efficient,kurita2020efficient,Cao:arXiv:2020} have appeared that apply the technique to obtain new output-polynomial algorithms. 
In particular,~\cite{Cao:arXiv:2020} solves the enumeration of Maximal Induced Interval Subgraphs by proposing some variations to proximity search~\cite{Conte2019proximity}.

\subsection{Overview}
The main contribution of the paper is presenting \textit{proximity search}, a general technique that can be used to solve several enumeration problems in polynomial delay, and \textit{canonical reconstruction}, a way to design a proximity search algorithm by exploiting orderings of solutions of the problem at hand.
exponentially many solutions.


By using this technique we show polynomial delay algorithms for several maximal listing problems such as maximal bipartite subgraphs and the others mentioned in Theorem~\ref{thm:main}. Other than providing efficient algorithms, we remark that the technique may help gain further insight on which classes of problems allow output-polynomial listing algorithms and which do not.

The paper is organized as follows:
First, we introduce some basic concepts and notation in Section~\ref{sec:prelim}. 
%
%
We then explain the proximity search technique, and formally define a class of problems, called \textit{proximity searchable}, which allow for a polynomial delay algorithm by its application.

Generality comes sometimes at the expense of efficiency but allows for a more intuitive understanding of the concepts at hand. For this reason, we divide the explanation in two parts: the first one, in Section~\ref{sec:outline}, formalizes the constraints required for a proximity search algorithm. The second, Section~\ref{sec:reconstruction}, introduces a technique which we call \textit{canonical reconstruction} for implementing proximity search. While canonical reconstruction is not the only way to obtain a proximity search algorithm, we observed that is often a powerful and elegant way to model the problem at hand.

Following, Sections~\ref{sec:bip}-\ref{sec:list:last}, shows how to prove that the problems in Theorem~\ref{thm:main} are proximity searchable and thus allow polynomial-delay algorithms.

As a drawback of the above algorithms is an exponential space requirement, we then propose a technique to address this issue, when suitable conditions are met: define a parent-child relation between solutions, in the style of reverse-search, as detailed in Section~\ref{sec:pspace-expl}, and give the algorithms in Section~\ref{sec:pspace-algs}. The resulting bounds are shown in Theorem~\ref{thm:pspace}. While this technique does not apply to all problems in Theorem~\ref{thm:main}, when it does it allows us to obtain polynomial-delay and polynomial-space algorithms for several problems whose input-restricted problem cannot be solved in polynomial time, including non-hereditary ones.

\section{Preliminaries}\label{sec:prelim}
Most of the enumeration problems addressed in this paper consider a simple undirected graph $G$, whose vertex set is denoted as $V(G)$ and edge set as $E(G)$, or simply $G=(V,E)$ when it is clear from the context. The neighborhood of a vertex $v$ is denoted as $N(v)$. For brevity, we refer to $|V(G)|$ as the number $n$ of vertices, to $|E(G)|$ as the number $m$ of edges, and to the maximum degree of a vertex in $G$ as $\Delta = \max_{v \in V} |N(v)|$. Furthermore, we assume the vertices to be labeled arbitrarily in increasing order $v_1, \ldots, v_n$, and say that $v_i$ is 
smaller than $v_j$ if $i<j$. 
We say that a neighbor of $v_i$ is a \textit{forward} neighbor if it comes later than $v_i$ in the order, and a \textit{backward} neighbor otherwise.
%
%

For a set of vertices $A\subseteq V(G)$, $E[A]$ denotes the edges of $G$ whose endpoints are both in $A$, and $G[A]$ the graph $(A, E[A])$, i.e., the subgraph induced in $G$ by $A$. Similarly, for a set $B$ of edges, $V[B]$ denotes the vertices incident to an edge in $B$ and $G[B] = (V[B], B)$. As common in the literature, we call \textit{induced subgraphs} those of the former kind, defined by a set of vertices, and \textit{edge-induced subgraphs} (or simply subgraphs) those of the latter, defined by a set of edges.
When dealing with subgraphs defined by a set of vertices (resp. edges) $A$, we will sometimes use $A$ to refer to both the vertex set (resp. edge set) and the subgraph $G[A]$ it induces, when this causes no ambiguity. We will also use $\ccomp_v(A)$ to refer to the \textit{connected component} of $G[A]$ which includes the vertex $v$. For further notation, we refer to the standard terminology in~\cite{Diestel2005}.



For a set of vertices $A\subseteq V(G)$ which corresponds to a solution of the problem at hand, we say that $A$ is \emph{maximal} if there is no $A' \subseteq V(G)$ such that $A' \supset A$  and $A'$ is also a solution. 
While not strictly necessary for the proximity search technique, in the following we will often rely on a simple ``maximalization'' function, named $\comp(A)$: this function takes a (not necessarily maximal) solution $A$ and ``completes'' it, returning some maximal solution $A'\supseteq A$. We will refer to the computational cost of this function as $\ccost$. 
Note that it is always possible to devise a polynomial-time computable $\comp(\cdot)$ function for hereditary and connected-hereditary properties where solutions can be recognized in polynomial time, by simply trying to add vertices until no longer possible~\cite{Cohen20081147}.

%

For simplicity, we disregard the presence of isolated vertices in the complexity analysis of the algorithms provided: these are trivially handled for the problems considered in this paper (either they can all be added ``in bulk'' to every solution, as for bipartite subgraphs, or each constitutes a maximal solution by itself, as for connected bipartite subgraphs), can be removed with an $O(n)$ time preprocessing; this means we are able to perform operations like a visit of the graph in $O(m)$ time rather than $O(m+n)$ time.

\section{Proximity search outline}\label{sec:outline}

Proximity search is based on traversing an implicit solution graph, where the vertices are all the solutions to be listed and each directed arc goes from a solution to another using a neighboring function. Several solution graphs are possible, depending on how the neighboring function is defined. Apart from the fact that the resulting solution graphs are not necessarily strongly connected and some care should be taken to list all the solutions, the main hurdle is that the degree of the solution graphs can be exponential (as the number of solutions can be exponentially large in the input size), thus preventing to achieve polynomial delay when running a simple traversal. Proximity search circumvents these issues by designing a suitable neighboring function, denoted $\neighs(\cdot)$, that guarantees that the resulting solution graph it implicitly defines is strongly connected and of polynomial degree. Both these properties cannot be guaranteed
with the current state of the art for a number of problems discussed later. 

We devote this section to formalize the general structure of proximity search, and the class of problems to which the technique can be applied. Also, we introduce the notion of proximity, symbolized by $\scap$, to act as a sort of oracle for navigating the solution graph.

For reference in what we discuss next, we give the pseudo-code of the generic traversal of a solution graph based on the $\neighs(\cdot)$ function, as shown in Algorithm~\ref{alg:general}. As noted earlier, the algorithms obtained by specializing this generic traversal are remarkably simple: In a depth-first search traversal where the set \sol keeps track of just the last visited solution, we only need to implement the $\neighs(\cdot)$ function. 
On the other hand, the complexity is mostly hidden behind proving their completeness: Notably, the very notion of proximity $\scap$ is only used in the proofs, and never actually appears in Algorithm~\ref{alg:general}.

\begin{algorithm2e}[ht]
\DontPrintSemicolon
\SetKwInOut{Input}{input}
\SetKwInOut{Output}{output}
\SetKwInOut{Global}{global}

\Input{Graph $G = (V,E)$ and listing problem $\mathcal{P}$}
\Output{All (maximal) solutions of $\mathcal{P}$ in $G$}
\Global{Set $\sol$ of solutions found, initially empty}

\BlankLine

$S \gets$ an arbitrary solution of $\mathcal{P}$ \label{ln:g:first}\tcc{ (e.g. $\comp(\emptyset)$)}
Call $\rec(S)$\;

\BlankLine


\SetKwProg{myproc}{Function}{}{}

\myproc{$\rec(S)$}{
    
    Add $S$ to $\sol$\label{ln:g:add}\;
    
    \tcc{\small Output $S$ if recursion depth is even}
    
    
    \ForEach{$S'\in \neighs(S)$\label{ln:g:neighs}}{
    
    \lIf{$S' \not \in\sol$}{$\rec(S')$\label{ln:g:rec}}
    }
    
    \tcc{\small Output $S$ if recursion depth is odd}
}
\caption{Traversal of the solution graph by proximity search.}\label{alg:general}
\end{algorithm2e}

In order to start the algorithm, we need one arbitrary maximal solution $S$. We remark that identifying one maximal (not maximum) solution is typically trivial, and can be achieved for example
by running $\comp(\emptyset)$
when the $\comp(\cdot)$ function is computable in polynomial time.

We formally define the class of problems which allow for a polynomial delay algorithm using this structure as \textit{proximity searchable}.

\begin{definition}[Proximity searchable]\label{def:searchable}
Let $\mathcal{P}$ be a listing problem over a universe $\mathcal{U}$ with set of solutions $\mathcal{S}\subseteq 2^\mathcal{U}$, where each solution is a subset of the universe. $\mathcal{P}$ is \textit{proximity searchable} if there exists a proximity function $\scap : \mathcal{S}\times \mathcal{S}\rightarrow 2^\mathcal{U}$ and a neighboring function $\neighs(\cdot) : \mathcal{S} \rightarrow 2^{\mathcal{S}}$,
such that the following holds: 

\begin{enumerate}
    \item \label{item:prox1} One solution of $\mathcal{P}$ can be identified in time polynomial in $|\mathcal{U}|$.
    \item \label{item:prox2} $\neighs(\cdot)$ is computable in time polynomial in $|\mathcal{U}|$.
    \item \label{item:prox3} Given any two distinct solutions $S,S^*\in \mathcal{S}$, there exists $S' \in \neighs(S)$ such that $|S'\scap S^*| > |S\scap S^*|$.
\end{enumerate}
The above conditions imply the following one, which is reported for the sake of clarity.
\begin{enumerate}
    \setcounter{enumi}{3}
    \item \label{item:prox4} For any fixed $S^*$, $|S\scap S^*|$ is maximized for (and only for) $S=S^*$.
\end{enumerate}

\end{definition}

If a problem is proximity searchable, then it is straightforward to see that we obtain a polynomial delay algorithm for it by using the corresponding $\neighs(\cdot)$ function in Algorithm~\ref{alg:general}. Let us formally prove it.

\begin{theorem}
\label{thm:strong}
All proximity searchable listing problems have a polynomial delay listing algorithm.
\end{theorem}
\begin{proof}
We first show that if a $\neighs(\cdot)$ function satisfies Definition~\ref{def:searchable}, the implicit solution graph it induces is strongly connected. Given any two distinct solutions $S, S^* \in \mathcal{S}$, we know by Definition~\ref{def:searchable}.(\ref{item:prox3}) that there exists $S'\in\neighs(S)$ such that $|S'\scap S^*| > |S\scap S^*|$. By induction on $S$ and $S'$, it follows that we will eventually reach a solution $S$ that globally maximizes $|\cdot \scap S^*|$, which by Definition~\ref{def:searchable}.(\ref{item:prox4}) is precisely $S^*$. 

Based on the above properties, we next show that Algorithm~\ref{alg:general} outputs (all and only) the solutions of any proximity searchable problem with no duplication. 
Firstly, Algorithm~\ref{alg:general} returns only maximal solutions, as it only outputs the initial maximal solution found on line~\ref{ln:g:first}, found polynomial time by Definition~\ref{def:searchable}.(\ref{item:prox1}), and the output of calls to $\neighs(S)$ which contain maximal solutions.

We say that a solution is \textit{visited} when $\rec(S)$ is called. In Algorithm~\ref{alg:general} all solutions added to $\sol$ are visited at most once, thanks to the membership test in the set $\sol$; this guarantees that the same solution is never output twice. As the graph defined by $\neighs(S)$ is strongly connected, the traversal done by Algorithm~\ref{alg:general} starting from the solution found on line~\ref{ln:g:first} must find all solutions.

To complete the proof, we show that Algorithm~\ref{alg:general} runs in polynomial delay.

Firstly, $\neighs(S)$ requires polynomial time by Definition~\ref{def:searchable}.(\ref{item:prox2}) and thus can only return a polynomial number of solutions; this means the out-degree of every node in the implicit solution graph is polynomial and we can iterate over it in polynomial time.


As a new recursive call is performed only when a new solution is found, the amortized cost per solution is bound by the cost of a recursive call, i.e., the cost of lines~\ref{ln:g:add}--\ref{ln:g:rec}. As the cost of $\neighs(\cdot)$ is polynomial, and $\sol$ can be easily maintained in polynomial time (in Appendix~\ref{sec:sol} we show this latter cost to be negligible for all algorithms presented here), it follows that the amortized cost per solution is polynomial. In order to get polynomial delay, we can employ the \textit{alternative output}~\cite{Uno2003} method, that can be applied to any recursive algorithm that outputs a solution in each recursive call: by performing output in pre-order when the recursion depth is even, and post-order when it is odd, the delay will be bounded by that of a constant number of recursive calls, i.e., polynomial.
\end{proof}

The following observations are in order:
\begin{itemize} 
    \item $\scap$ is not actually used in Algorithm~\ref{alg:general}, and does \textit{not} need to be computed.
    \item Proximity search can be applied to all listing domains where solutions are modeled by set systems, not just graphs.
    \item Proximity search is mainly intended for maximal listing problems, however, it is not strictly limited to it.
    \item Maximal listing problems in which the input-restricted problem is computable in polynomial time (as well as the $\comp(\cdot)$ function) are proximity searchable.\footnote{In essence, we obtain as a special case the same solution graph as known algorithms based on the input-restricted problem~\cite{tsukiyama1977new,lawler1980generating,Cohen20081147}: For a solution $S$ we compute $\neighs(S)$ by sequentially taking all elements $v\in \mathcal{U}\setminus S$, solving the input-restricted problem for $S \cup \{v\}$, and applying $\comp(\cdot)$ on the results.}
    \item The polynomiality constraint on $\neighs(\cdot)$ can be relaxed: it can be trivially seen how computing $\neighs(\cdot)$ in Incremental Polynomial Time (resp. Polynomial Total Time) yields and Incremental Polynomial Time (resp. Polynomial Total Time) algorithm. 
    \item The cost per solution and delay of the algorithm is the complexity of the $\neighs(\cdot)$ function (we show in Appendix~\ref{sec:sol} how maintaining $\sol$ is negligible).
\end{itemize}

In the rest of the paper, we show how to suitably model several problems to obtain new polynomial-delay algorithms for problems that, to the best of our knowledge, could not be previously solved in polynomial delay. 
We show these algorithms by providing suitable $\scap$ and $\neighs(\cdot)$ functions, proving that they satisfy Definition~\ref{def:searchable}, which automatically give us a polynomial delay listing algorithm by Algorithm~\ref{alg:general}.
We will use a common notation: $S$ is an arbitrary solution, and $S^*$ the ``target'' solution.





\begin{figure*}[t]
    \centering
    \includegraphics[width=\textwidth]{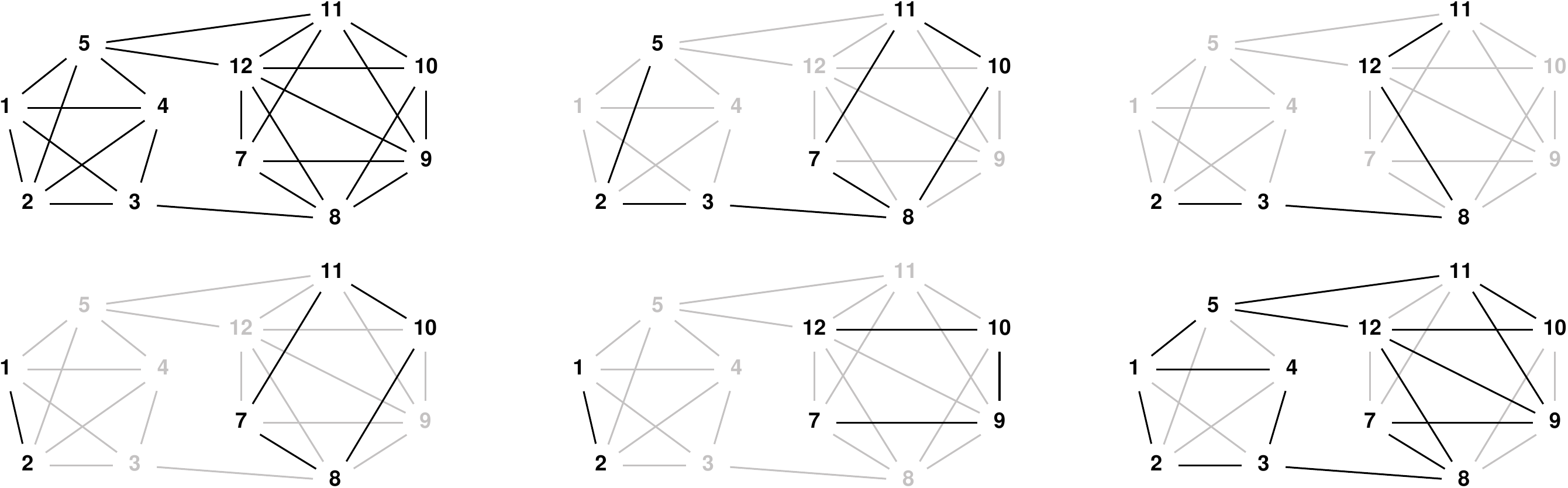}
    \put (-505,145){\small(a)}
    \put (-320,145){\small(b)}
    \put (-140,145){\small(c)}
    \put (-505,60){\small(d)}
    \put (-320,60){\small(e)}
    \put (-140,60){\small(f)}
    %
    %
    \caption{\textbf{a}: a graph. \textbf{b},\textbf{c}: two maximal connected induced bipartite subgraphs of (a). \textbf{d},\textbf{e}: two maximal induced bipartite subgraphs of (a). \textbf{f}: a maximal edge-induced bipartite subgraph of (a).}
    \label{fig:ex:bip}
\end{figure*}
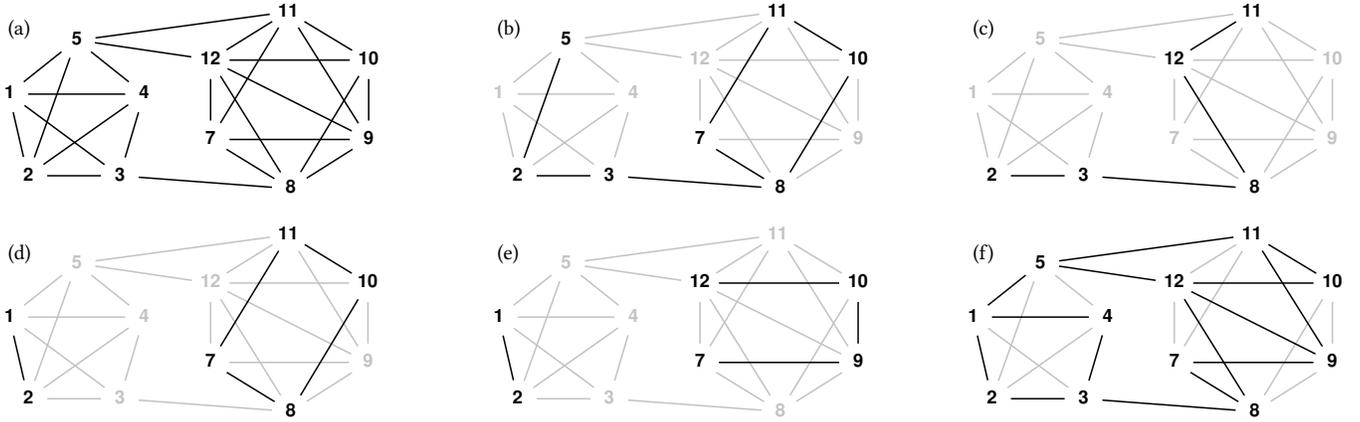

\section{Proximity search by canonical reconstruction}\label{sec:reconstruction}


We make concrete use of the abstract notion of proximity search and introduce a technique, which we call \textit{canonical reconstruction}.
While it is kept separate from the previous section for cohesiveness, we find this 
technique to often be the right way to look at maximal subgraph listing problems. 
Since we deal with graphs, the universe $\mathcal{U}$ is the vertex set, unless explicitly specified.

To accompany the explanation, we detail its implementation in the case of Maximal Connected Induced Bipartite Subgraphs in Section~\ref{sec:bip}.


The technique is based on the definitions of  \textit{canonical order} and \textit{canonical extender} for solutions, which depend entirely on the problem at hand, and it is intuitively a way to harness its structure.

\begin{example*}
For a Maximal Connected Induced Bipartite Subgraph $S$, we will use as canonical order a BFS-order of $G[S]$ starting from its vertex of smallest id, where ties are broken by vertex id: 
For the subgraph in Figure~\ref{fig:ex:bip} (b) this order is\linebreak $2,3,5,8,11,7,10$, and for the one in (c) it is $2,3,8,12,11$.
\end{example*}

\paragraph{Canonical order and proximity} 
Simply assume that each solution $S$ is given an ordering $s_1,\ldots, s_{|S|}$ of its elements which will satisfy some problem-specific conditions. We require that any \textit{prefix} $s_1,\ldots, s_i$ of this order corresponds to a (non-maximal) solution $\{ s_1,\ldots, s_i\}$. In the rest of the paper, we will refer as prefix of the order to both the sequence $s_1,\ldots, s_i$ and the corresponding set of elements $\{ s_1,\ldots, s_i\}$.
Note that the ordering is \emph{not} required to be efficiently computable, as the proximity search algorithm never actually computes it: it is only used in the correctness proof of the neighboring function. Moreover, the ordering is adaptive to each solution, so the same elements can be ranked differently in distinct solutions.

Given the order, we define the proximity function $\scap$ as follows.
\begin{definition}[proximity]\label{def:proximity}
Given two solutions $S$ and $S^*$, let $ s^*_1,\ldots s^*_{|S^*|}$ be the canonical order of $S^*$: the proximity $S \scap S^*$ between $S$ and $S^*$ is the longest \emph{prefix} $ s^*_1,\ldots,s^*_i$ of the canonical order of $S^*$ whose elements are all contained in $S$.
\end{definition}

It should also be noted that the operation is \textit{not} symmetric, i.e., we may have $S \scap S^* \ne S^* \scap S$. 

\begin{example*}
Let $S$ be the subgraph shown in Figure~\ref{fig:ex:bip} (b) and $S^*$ the one shown in (c). Considering the canonical orders mentioned above, we can see that $S \scap S^* = \{2,3,8\}$, while $S^* \scap S = \{2,3\}$.
\end{example*}

\paragraph{Canonical extender}
The goal of a proximity search algorithm is to exploit Definition~\ref{def:searchable}.(\ref{item:prox3}): given $S$, for any $S^*$, find some $S'$ such that $|S'\scap S^*|>|S\scap S^*|$. Using Definition~\ref{def:proximity}, $S\scap S^*$ is a prefix  $s^*_1,\ldots,s^*_i$ of the canonical order of $S^*$, so we want to find any solution $S'$ that contains a longer prefix, i.e., $s^*_1,\ldots,s^*_{i+1}$ (possibly ordered differently and interspersed in the canonical order of $S'$). Since we must at least add the vertex $s^*_{i+1}$, we call $s^*_{i+1}$ the \textit{canonical extender} of $S,S^*$. Armed with this notion, we want to proceed conceptually as follows for a given solution $S$.
\begin{enumerate}
    \item Guess which node $v \not \in S$ is the canonical extender $s^*_{i+1}$ (try all possibilities, $n$ at most).
    \item Guess a \textit{removable set} $X\subseteq S$ from $S\cup\{v\}$, i.e., such that $S\setminus X\cup\{v\}$ is a solution and $X \cap \{s^*_1,\ldots,s^*_{i}\} = \emptyset$.
    \item Obtain $S'$ as the outcome of $\comp(S\setminus X\cup\{v\})$.
\end{enumerate}

In essence, we want to add $s^*_{i+1}$ to $S$, then turn the result back into a solution by removing some elements, but without affecting the proximity $s^*_1,\ldots,s^*_{i}$.

Recalling that prefixes of a canonical order are required to be (non-maximal) solutions, indeed $s^*_1,\ldots,s^*_{i+1}$ is a solution; hence, a removable set $X$ always exists (e.g., $X=S\setminus \{s^*_1,\ldots,s^*_{i}\}$). 
The key point is that we want to satisfy the proximity requirement for all $S^*$ (that can be exponentially many) using only a \textit{polynomial} number of removable sets $X$.
While there is no general rule for this, and indeed, solving this for some problems would imply \textsc{p}=\textsc{np}, we will observe in this paper how it is possible to do so in some cases where a canonical order can efficiently decompose the underlying structure of the solution.

\paragraph{Canonical reconstruction}
Now we have all the ingredients to formalize below the required structure for adopting our strategy.
\begin{definition}(Proximity search by canonical reconstruction)
\label{def:crecon}
Given a maximal listing problem $\mathcal{P}$, in which each maximal solution $S$ is associated with a canonical ordering $s_1,\ldots, s_{|S|}$, we say that $\mathcal{P}$ admits a \emph{canonical reconstruction} if the following holds.
\begin{enumerate}
    \item \label{item:canonical1} Any prefix $s_1,\ldots, s_i$ of the canonical order of any maximal solution $S$ is a (non-maximal) solution of $~\mathcal{P}$.
    \item \label{item:canonical2} Given a maximal solution $S$ and any $v\not\in S$, there is set $\mathcal{X}\subseteq 2^S$ of \emph{removables}, such that 
    \begin{itemize} 
    	\item $\mathcal{X} = \{ X_1, X_2,~\ldots \}$ can be computed in polynomial time.
        \item $S\setminus X_i\cup\{v\}$ is a solution of $\mathcal{P}$ for any $X_i\in \mathcal{X}$.
        \item For any $S^*$ such that $v$ is the canonical extender of $S,S^*$, there is at least one $X_i\in \mathcal{X}$ such that $(S\scap S^*)  \cap X_i = \emptyset$.
        \footnote{Indeed, $(S\scap S^*)  \cap X_i = \emptyset$, i.e., $X_i$ does not intersect the proximity, implies that $S\setminus X_i\cup\{v\}$ contains $(S\scap S^*) \cup \{v\}$, which extends the proximity with $v$.}
    \end{itemize}
    \item There is a polynomial-time computable function $\comp(A)$ which takes a solution $A$ of $~\mathcal{P}$ and returns a maximal solution $A'\supseteq A$ of $~\mathcal{P}$.
\end{enumerate}

We further define the \emph{canonical reconstruction function}, as 
\[
\neighs(S,v) = \bigcup_{X_i\in\mathcal{X}}\comp((S\setminus X_i)\cup\{v\})
\]
and this corresponds to the solutions $S' \in \neighs(S)$ for which $v$ is the canonical extender of $S,S'$.
Hence, $\neighs(S)$ is obtained as $\bigcup_{v \in V(G)} \neighs(S,v)$.\footnote{For completeness, we define $\neighs(S,v) = \{S\}$ for $v\in S$, as $S\cup \{v\} = S$ is already a solution of $\mathcal{P}$.} 
\end{definition}

We observe how the removables and the neighboring function can be derived from one another, so the algorithm can be defined by providing either one: one can focus on defining removables from $S\cup\{v\}$ that do not intersect the proximity, or equivalently solutions contained in $S\cup\{v\}$ that fully contain the proximity.

We also recall that a polynomial-time computable function $\comp(\cdot)$ trivially exists for all hereditary and connected-hereditary properties that can be recognized in polynomial time (these include all the example problems shown in this paper).

\begin{figure*}
    \centering
    \includegraphics[width=\textwidth]{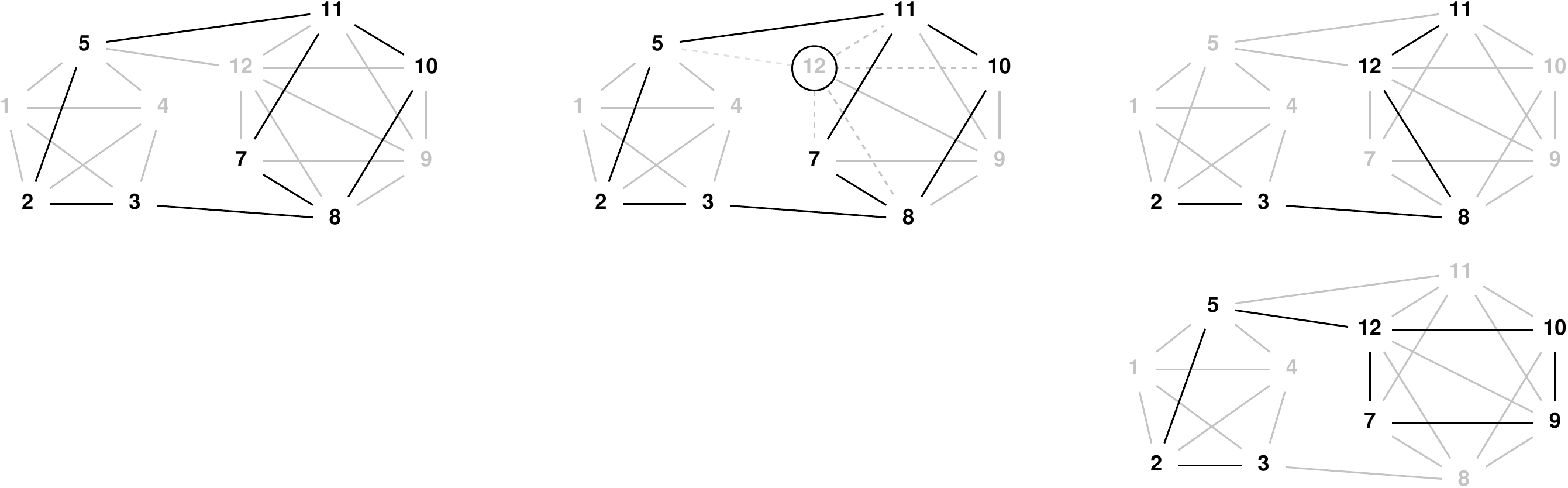}
    \put(-470,50){\begin{scriptsize}$S = \{2,3,5,7,8,10,11\}$, with $S_0 = \{2,8,11\}$, $S_1 = \{3,5,7,10\}$, and $v = 12$\end{scriptsize}}
    \put(-470,33){\begin{scriptsize}$X_0 = N(v)\cap S_0 = \{8,11\}$,\quad $X_1 = N(v)\cap S_1 = \{5,7,10\}$\end{scriptsize} }
    \put(-470,16){\begin{scriptsize}$S \setminus X_1 \cup \{v\} = \{2,3,8,11,12\}$ (top)\end{scriptsize}}
    \put(-470,0){\begin{scriptsize}$S \setminus X_0 \cup \{v\} = \{2,3,5,7,10,12\}$, with vertex $9$ added by $\comp(\cdot)$ (bottom)\end{scriptsize} }
    \caption{The steps taken by the neighboring function $\neighs(S,v)$, for a possible solution $S$ of the graph in Figure~\ref{fig:ex:bip} and $v=12$. The bottom two lines correspond to the neighboring solutions produced by the function $\neighs(S,12)$.}
    \label{fig:bip-running}
\end{figure*}
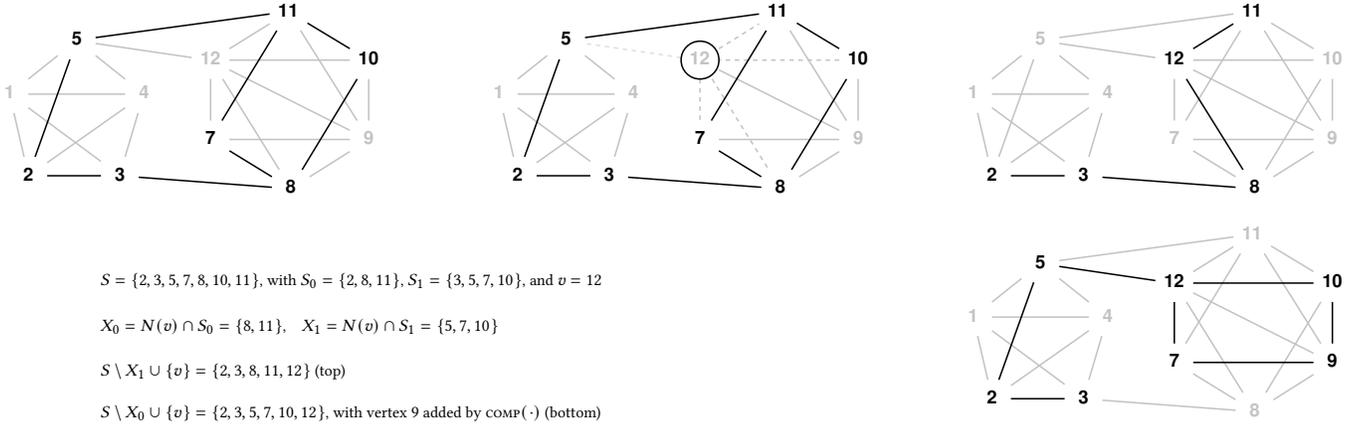

Finally, we show how canonical reconstruction immediately implies the maximal listing problem at hand is proximity searchable, using the $\neighs(\cdot)$ function defined above:

\begin{theorem}
\label{thm:canonical-reconstructio}
All maximal listing problems that allow a canonical reconstruction are proximity searchable.
\end{theorem}
\begin{proof}
Let us show that a listing problem $\mathcal{P}$ that satisfies Definition~\ref{def:crecon} satisfies the four conditions of Definition~\ref{def:searchable}.
Condition~(\ref{item:prox1}) is trivially satisfied, say, using $\comp(\emptyset)$.
As for condition~(\ref{item:prox2}), recall that 
\[
\neighs(S) = 
\bigcup_{v \in V(G)} \bigcup_{X_i\in\mathcal{X}}(\comp(S\setminus X_i\cup\{v\})).
\]
Considering that both $|V(G)|$ and $|\mathcal{X}|$ are polynomial, and $\comp(\cdot)$ takes polynomial time, it follows that computing $\neighs(S)$ takes polynomial time.

For condition~(\ref{item:prox3}), consider the canonical extender $v$ for $S,S^*$: By Definition~\ref{def:crecon} there is $X_i\in \mathcal{X}$ such that $X_i\cap (S\scap S^*) = \emptyset$; it holds that $S' = \comp(S\setminus X_i\cup\{v\}) \in \neighs(S,v) \subseteq \neighs(S)$, and $(S\scap S^*) \cup \{v\}\subseteq S'$, thus $|S'\scap S^*|>|S\scap S^*|$ because of $v$. 
  
Finally, condition~(\ref{item:prox4}) is satisfied by looking at the definition of proximity in Definition~\ref{def:proximity}: fixed $S^*$, the proximity $S\scap S^*$ is maximized if $S \supseteq S^*$, as we have $S\scap S^* = S^* = S$; however, as this is a maximal listing problem, all solutions are inclusion-wise maximal, meaning that $S \supseteq S^*$ is only true for $S=S^*$. 
\end{proof}

As a final remark, we note a somewhat surprising feature of this technique: while in general connected-hereditary properties (e.g., Maximal Connected Induced Bipartite Subgraphs) are more challenging to deal with than hereditary ones (e.g., Maximal Induced Bipartite Subgraphs), in the case of proximity search there is typically no difference, and in some instances, we even use the connected case as a starting point for the non-connected one (see, e.g., Section~\ref{sec:binon}).


\section{Maximal Bipartite Subgraphs}\label{sec:bip}

We now illustrate how to apply proximity search to maximal bipartite subgraph enumeration, giving the full details for the example in the previous section. 

A graph $G$ is bipartite if its vertices can be partitioned into two sets $V_0,V_1$, such that $V_0\cap V_1 = \emptyset$, $V_0 \cup V_1 = V(G)$, and both $G[V_0]$ and $G[V_1]$ are edge-less graphs. Equivalently, $G$ is bipartite if it has no cycle of odd length.
Maximal bipartite subgraphs have also been studied as \textit{minimal odd cycle transversals}~\cite{kratsch2014compression}, as one is the complement of the other.

The problem of listing \textit{all} bipartite (and induced bipartite) subgraphs has been efficiently solved in~\cite{wasa2018bipartite}. However, to the best of our knowledge, neither the techniques in~\cite{wasa2018bipartite} nor other known ones extend to efficiently listing \textit{maximal} bipartite subgraphs, which poses a challenge. Consider the instance of input-restricted problem shown in Figure~\ref{fig:bip:restr} (left). We can exploit the fact that a subgraph of a bipartite graph is itself bipartite, meaning that the property is hereditary.
Hence, we could take the current solution $S$ (which are the endpoints of the bold edges) and a vertex $v \not \in S$, to then try to list all the maximal solutions contained in the induced subgraph $G[S \cup \{v\}]$; however, $G[S \cup \{v\}]$ has exponentially many solutions, meaning we cannot solve the input-restricted problem in polynomial time and thus we cannot get polynomial delay with the techniques from~\cite{lawler1980generating,Cohen20081147,conte2019framework}. 
The best we could hope for is solving the input-restricted problem in polynomial delay or incremental polynomial time, which would yield an incremental polynomial time algorithm for the general problem~\cite{Cohen20081147}. Figure~\ref{fig:bip:restr} (right) shows an analogous situation for edge-induced subgraphs. 

We thus turn to proximity search. First, let us introduce some preliminary notions: We denote an induced bipartite \textit{subgraph} of $G$ as a pair of vertex sets $\langle B_0,B_1 \rangle$, with $B_0\cap B_1 = \emptyset$ and $B_0 \cup B_1 \subseteq V(G)$, such that $G[B_0]$ and $G[B_1]$ are edge-less graphs. By convention, $B_0$ is the side of the bipartition containing the vertex of smallest label among those in the subgraph. In case $G[B_0 \cup B_1]$ has multiple connected components, this applies to all components. This way, any bipartite subgraph (connected or not) always has a \emph{unique} representation $\langle B_0,B_1 \rangle$.
We will sometimes use simply $B$ to refer to the subgraph $G[B_0 \cup B_1]$ induced by $\langle B_0,B_1 \rangle$. When performing $\comp(B)$ (defined at the end of Section~\ref{sec:prelim}) and $B$ is not connected, this may move some vertices from $B_0$ to $B_1$ and vice versa due to different components becoming connected; even when $B$ is connected, if a vertex with smaller label than all others in $B$ is added to $B_1$, then $B_0$ and $B_1$ are immediately swapped to preserve the invariant of the smallest vertex being in $B_0$.
We define the intersection between two bipartite subgraphs $B$ and $B'$ as the set of all shared vertices, i.e.: $B\cap B' = (B_0 \cup B_1) \cap (B'_0\cup B'_1)$.

We consider the case of connected induced bipartite subgraphs in Section~\ref{sec:bicon}. We will later briefly show how this structure can be adapted to cover the non-connected and non-induced cases with small changes in Sections~\ref{sec:binon} and~\ref{sec:ebi}. Their complexity will be discussed in Section~\ref{sec:bitime}.

\subsection{Listing Maximal Connected Induced Bipartite Subgraphs}
\label{sec:bicon}

Let $B = \langle B_0,B_1 \rangle$ be a maximal induced bipartite subgraph of $G$, and $v$ a vertex not in $B$, i.e., in $v \in V(G)\setminus B$. 
Looking at Definition~\ref{def:crecon}, we immediately observe that a polynomial-time computable function $\comp(\cdot)$ exists since the problem is connected-hereditary.
Then, we need to define a suitable canonical order, and prove the existence of the corresponding removables. 

Consider a BFS order of $G[B]$ starting from its vertex of smallest label, say $b_1$. In this order, a vertex $u$ precedes a vertex $v$ if the distance of $u$ from $b_1$ is smaller than that of $v$ or, in case the distance is equal, $u$'s label is smaller than $v$'s. 

\begin{definition}[canonical order for connected induced bipartite subgraphs]\label{def:bip-con-canon} 
The \emph{canonical} order of a \emph{connected induced bipartite subgraph} $B$ is the sequence $b_1, \ldots, b_{|B|}$ given by a BFS order of $G[B]$ rooted at the vertex $b_1$ of smallest label, where ties are broken by placing the vertex of smallest label first.
\end{definition}

For the subgraph in Figure~\ref{fig:ex:bip} (b) the canonical order is $2,3,5,8,11,7,10$, and for the one in (c) it is $2,3,8,12,11$. The definition of proximity is then automatically given by Definition~\ref{def:proximity}.

Last ingredient for Definition~\ref{def:crecon} is the set $\mathcal{X}\subseteq 2^B$, that contains just \emph{two} removables. 
In order to get a bipartite graph, it is possible to make two removables as follows:
\begin{itemize}
    \item $X_0 = N(v) \cap B_0 $,
    \item $X_1 = N(v) \cap B_1 $.
\end{itemize}

That is, remove all the neighbors of $v$ in one of the two sides $B_i$: clearly, $v$ can be included in $B_i$ as it is now only adjacent to vertices of $B_{1-i}$.
While this works for the Maximal Induced Bipartite Subgraphs problem, we have the further constraint of \textit{connectivity}, so we must also discard every vertex that is not in the same connected component as $v$. The removables become as follows:

\begin{itemize}
    \item $X_0 = B \setminus \ccomp_v( \{v\} \cup (B_0 \setminus N(v)) \cup B_{1} )$,
    \item $X_1 = B \setminus \ccomp_v( \{v\} \cup (B_1 \setminus N(v)) \cup B_{0} )$.
\end{itemize}

That is, we remove all vertices not in the same connected component as $v$, after introducing $v$ and removing all its neighbors in either $B_0$ or $B_1$.

We can use these to create the neighboring function to be plugged in Algorithm~\ref{alg:general}, following Definition~\ref{def:crecon}:

\begin{definition}[neighboring function for maximal connected induced bipartite subgraphs]\label{def:neigh-bip-con}
$$\neighs(B,v) = \{ \comp( \ccomp_{v} ( \{v\} \cup (B_i \setminus N(v)) \cup B_{1-i} ) ) \mid i =0,1 \}$$

\end{definition}

A graphical example of this procedure is given in Figure~\ref{fig:bip-running} (for simplicity, we adopt an example where the subgraphs are connected after removing $v$'s neighbors, so the removables are equivalent to the simpler ones of the non-connected version).


\begin{lemma}\label{lem:bip-con-scap}
The problem of listing all Maximal Connected Induced Bipartite Subgraphs admits a canonical reconstruction.
\end{lemma}
\begin{proof}
For the canonical order given in Definition~\ref{def:bip-con-canon}, any prefix induces a graph that is connected because of the BFS order, and bipartite because bipartite subgraphs are hereditary, so condition~(\ref{item:canonical1}) of Definition~\ref{def:crecon} is satisfied. 
As for condition~(\ref{item:canonical2}), it is evident from the definition of removables (alternatively, of the neighboring function) that they can be computed in polynomial time, and that they produce connected bipartite subgraphs.
We only need to show that the third item holds: given $B$, $B^*$ and their canonical extender $\dv$, we have $B\scap B^*\cap X_i = \emptyset$ for either $i=0$ or $i=1$; this will imply that $|B' \scap B^*| > |B \scap B^*|$ for $B' = \comp( \ccomp_{\dv} ( \{\dv\} \cup (B_i \setminus N(\dv)) \cup B_{1-i} ) )$, so the proximity is successfully increased. 

If $B\scap B^* = \emptyset$ the claim is trivially true, as we can consider $b^*_1$ as canonical extender.
%
%
Otherwise, let $Z = B \scap B^* = \{b^*_1, \ldots, b^*_h\}$, and we have $\dv = b^*_{h+1}$. By Definition~\ref{def:bip-con-canon}, $Z$ is a connected induced bipartite subgraph, meaning that it allows a unique bipartition $Z_0,Z_1$ (with $Z_0$ being the set containing the vertex of smallest label in $Z$, that is, $b^*_1$). 
Since $b^*_1$ is the vertex of smallest label in $B^*$, it will be in $B^*_0$, so it follows that $Z_0 \subseteq B^*_0$ and $Z_1 \subseteq B^*_1$.

Let $j$ be the value in $\{0,1\}$ such that $\dv\in B^*_j$, and observe that $N(\dv)\cap Z_j \subseteq N(\dv)\cap B^*_j = \emptyset$.
Furthermore, we know that $b^*_1\in B^*_0$ and $b^*_1\in B$, but we do not know whether $b^*_1\in B_0$ or $b^*_1\in B_0$; however, there exists a value $i$ in $\{0,1\}$ such that either $b^*_1 \in (B_i\cap B^*_j)$ or $b^*_1 \in (B_{1-i}\cap B^*_{1-j})$. Observe that $Z_j \subseteq B^*_j \cap B_i$ and $Z_{1-j} \subseteq B^*_{1-j} \cap B_{1-i}$.

Finally, let $B' = \ccomp_{\dv} ( \{\dv\} \cup (B_i \setminus N(\dv)) \cup B_{1-i} )$, and consequently we have $X_i = B \setminus B'$.
$Z$ is fully contained in $B'$: the only vertices removed from $B$ by $X_i$ are (i) those in $N(\dv)\cap B_i$, but $N(\dv)\cap B_i \cap Z \subseteq N(\dv)\cap B_i \cap Z_j \subseteq N(\dv)\cap Z_j = \emptyset$, and (ii) the vertices not in the connected component of $\dv$ in $G[\{\dv\} \cup (B_i\setminus N(\dv)) \cup B_{1-i}]$, but no such vertex can be in $Z$ as $Z\cup \{\dv\}$ is a prefix of the canonical order of $B^*$, so it induces a connected subgraph.

We thus have that $Z\cup \{\dv\} \subseteq B'$, meaning that $X_i \cap Z = \emptyset$, which proves the claim.
Constructively, we can finally observe how the maximal solution $B''=\comp(B')$ is the one produces by the algorithm which increases the proximity to $B^*$, as we have $\{b^*_1, \ldots, b^*_h, b^*_{h+1}\} \subseteq B''\scap B^*$ and thus $|B'' \scap B^*| \ge |B \scap B^*|+1$.
\end{proof}


From this, we immediately obtain the correctness of the algorithm. 

\begin{theorem}\label{thm:bip-con}
A proximity search algorithm (Algorithm~\ref{alg:general}), using the\linebreak $\neighs(\cdot)$ function from Definition~\ref{def:neigh-bip-con} outputs all Maximal Connected Induced Bipartite Subgraphs of a graph $G$ without duplication with $O(nm)$ delay.
\end{theorem}
\begin{proof}
The correctness follows from Theorem~\ref{thm:canonical-reconstructio}.
The delay is dominated by the cost of the $\neighs(B)$ function, i.e., calling $O(n)$ times $\neighs(B,v)$. The cost of the latter is $O(m)$ time to compute $\ccomp_v(\cdot)$, and $O(m)$ time to compute the $\comp(\cdot)$ function by Lemma~\ref{lem:bip:comp} (delayed to Section~\ref{sec:bitime} for compactness). The statement follows.
\end{proof}

\subsection{Listing Maximal Induced Bipartite Subgraphs}
\label{sec:binon}

We can extend our solution to the non-connected case by building one connected component at a time.  We obtain the canonical order by Definition~\ref{def:canonBFSnon}, that is, a BFS order of each component:

\begin{definition}[canonical order for induced bipartite subgraphs]\label{def:bip-non-canon}
The \emph{canonical} order of an \emph{induced bipartite subgraph} $B$ is the sequence $b_1, \ldots, b_{|B|}$ obtained by first ordering the connected components of $G[B]$ by incremental order of smallest-id vertex, then ordering each component by a BFS order (given in Definition~\ref{def:bip-con-canon}) rooted in its smallest-id vertex.
\end{definition}

In essence, this corresponds to ordering each connected component as in the connected case (Definition~\ref{def:neigh-bip-con}), and placing earlier components whose smallest-id vertex is smaller.
Looking again at Figure~\ref{fig:ex:bip}, and letting $B$ be the subgraph shown in (d) and $B^*$ as that shown in (e), the canonical order of $B$ is $\langle 1,2,7,8,11,10\rangle $, that of $B^*$ is $\langle 1,2,7,9,12,10\rangle $. By the definition of proximity for canonical reconstruction, we also obtain $B\scap B^* = \{1,2,7\}$.

The removables become simpler for this case, as we can simply remove $N(v)\cap B_i$ for $i=0,1$.
%
As a result, the neighboring function is essentially the same as the connected case (Definition~\ref{def:neigh-bip-con}), with minor changes as we do not require the connectivity:

\begin{definition}[neighboring function for maximal induced bipartite subgraphs]\label{def:neigh-bip-non}
$$\neighs(B,v) = \{  \comp(\{v\} \cup (B_i \setminus N(v)) \cup B_{1-i} ) \mid i=0,1\}$$

\end{definition}

We can then proceed to prove correctness and complexity of this case:

\begin{theorem}\label{thm:nbip}
A proximity search algorithm (Algorithm~\ref{alg:general}), using the\linebreak $\neighs(\cdot)$ function from Definition~\ref{def:neigh-bip-non} outputs all maximal induced bipartite subgraphs of a graph $G$ without duplication with $O(n(m+n\iack(n)))$ delay.
\end{theorem}
\begin{proof}
Consider the solutions $B$ and $B^*$. Let $b^*_1,\ldots,b^*_{|B^*|}$ be the canonical ordering of $B^*$ by Definition~\ref{def:neigh-bip-non}, $B\scap B^* = b^*_1,\ldots,b^*_{i}$, and $u=b^*_{i+1}$ the canonical extender for $B,B^*$. Let $C$ be the connected component of $B^*$ containing $u$. Since all the neighbors of $u$ in $B^*$ must be in its same connected component $C^{B^ *}_x$, and the neighbouring function (Definition~\ref{def:neigh-bip-non}) only removes neighbors of $u$ from $B$, the function may not remove from $B$ any vertex of $B^*$ that is \textit{not} in $C$. As for vertices in $C$, $B\scap B^*$ contains a (possibly empty) prefix of its BFS order, which is itself a connected bipartite subgraph in canonical order. By the correctness of Lemma~\ref{lem:bip-con-scap}, for either $B' = \comp( \ccomp_{v} ( \{v\} \cup (B_0 \setminus N(v)) \cup B_{1} ) ) $ or $B' = \comp( \ccomp_{v} ( \{v\} \cup (B_1 \setminus N(v)) \cup B_{0} ) )$, this prefix is expanded with $u$, giving us $B'\scap B^* \supseteq (B\scap B^*)\cup\{u\}$ and proving correctness.

As for the delay, we can see that the cost of $\neighs(B)$ is bounded as for the connected case by $O(n)$ times the cost of $\neighs(B,v)$, which is in turn bounded by $O(m + n\iack(n))$ by Lemma~\ref{lem:bip:comp}, proving the statement.


\end{proof}

\subsection{Maximal Edge Bipartite Subgraphs}
\label{sec:ebi}

Finally, we show how to adapt the above algorithm to Maximal \textit{Edge} Bipartite Subgraphs, where edge-induced subgraphs are denoted by a set of edges, rather than vertices.  In the following, given two sets of vertices $A$ and $B$, let $E(A,B)$ be the set of edges with one endpoint in $A$ and the other in $B$. We observe that the \textit{Maximal} Edge Bipartite Subgraphs of a connected graph are always connected, otherwise some edge could be added to joint components without creating cycles; by the same logic they span all vertices, and may thus be represented by simply a bipartition $\langle B_0,B_1 \rangle$ of $V(G)$, where the bipartite subgraph corresponds to the edges in $E(B_0,B_1)$. For readability, we use the shorthand $E_B \equiv E(B_0,B_{1})$ to refer to the edges of the bipartite subgraph $B$.

We also observe that the problem is hereditary and allows for a polynomial time computable $\comp(\cdot)$ function. 
We define the canonical order of a solution $B$ by taking the canonical order $ b_1, \ldots, b_{|B|} $  of the \textit{vertices} of $G[B]$ according to Definition~\ref{def:bip-con-canon},\footnote{Note that the vertices of $G[B]$ are all of $V(G)$, but to compute the canonical order we need to consider only the edges in the bipartite subgraph $G[B]$.} then taking the edges of $B$ in increasing order of their \textit{latter} vertex in the vertex order, and breaking ties by increasing order of the earlier endpoint. This essentially corresponds to ``building'' $B$ in a similar fashion as in the induced version, but adding one edge at a time incident to the newly selected vertex. The removables for an edge $e = \{a,b\}$, where $a < b$, are as follows.

\begin{itemize}
    \item $X_0 = (E_B \setminus N_E(a))\cup \{e\} )$,
    \item $X_1 = (E_B \setminus N_E(b))\cup \{e\} )$.
\end{itemize}

The principle behind the neighboring function is different but inspired by the induced case: rather than taking a vertex out of the solution and trying to add it to $B_0$ or $B_1$, we take an edge $e = \{a,b\}$ with both endpoints in the same $B_i$, and try to move the two vertices $a$ and $b$ to opposite sides of the bipartition.

This can be achieved by including the edge $e$ in the solution, and then, to preserve the subgraph being bipartite, removing either $N_E(a)$ or $N_E(b)$ from it. Finally we apply the $\comp(\cdot)$ function to obtain a solution that is maximal. 

More formally, recalling $E_B \equiv E(B_0,B_{1})$, we define $\neighs(B)$ as 

$$\bigcup\limits_{e=\{a,b\}\in E(G)\setminus E_B}\{
\comp((E_B \setminus N_E(a))\cup \{e\}),\; \comp((E_B \setminus N_E(b))\cup \{e\})
\}$$






Consider two solution $B$ and $B^*$, with $e_1, \ldots, e_{|E(B^*_0,B^*_1)|}$ being the canonical order of $B^*$. Furthermore, let $B\scap B^* = \{e_1, \ldots, e_h\}$ and $\de = e_{h+1} = \{a,b\}$ the canonical extender, i.e., the first edge in the ordering of $B^*$ which is not in $B$.

By the definition of the canonical ordering, we have that $\{e_1, \ldots, e_h\}$ is a connected bipartite subgraph, meaning that it allows a unique bipartition $B' = B'_0, B'_1$ of its incident vertices. As $\{e_1, \ldots, e_h\}\cup \{\de\}$ is also a connected bipartite subgraph, for some $j \in \{0,1\}$ we must have both $N(a)\cap B'_j = \emptyset$ and $N(b)\cap B'_{1-j} = \emptyset$.

Since $B'$ is included in $B$, we must have either (i) $B'_j\subseteq B_i$ and $B'_{1-j}\subseteq B_{1-i}$ or (ii) $B'_{1-j}\subseteq B_i$ and $B'_{j}\subseteq B_{1-i}$. 
Recall now that both $a$ and $b$ are assumed wlog to be in $B_i$, meaning that $N(a)\cap B_i = N(b)\cap B_i = \emptyset$.
In the (i) case, we have $N(b) \cap  B_{1-i} \cap B'_{1-j} = \emptyset$, so removing $N_E(b)$ from $B$ may not remove any edge of $B'$. 
Thus 
$\comp((E_B \setminus N_E(b))\cup \{\de\})$, which belongs to $\neighs(B)$,
will contain $(B\scap B^*)\cup\{e\}$.

In the (ii) case, we have $N(a) \cap  B_{1-i} \cap B'_{1-j} = \emptyset$, removing $N_E(a)$ may not remove any edge of $B'$. 
Thus 
$\comp((E_B \setminus N_E(a))\cup \{\de\})$, which also belongs to $\neighs(B)$,
will contain $(B\scap B^*)\cup\{e\}$.

This means that in both cases, $\neighs(B)$ will yield a solution $S'$ that contains $(B\scap B^*)\cup\{e\}$, i.e., such that $|S'\scap S^*| > |S \scap S^*|$.

As the complexity is bounded by $O(m)$ calls to the $\neighs(B_i,B_{1-i},e)$ function, whose cost is again bounded by that of $\comp(\cdot)$, that is $O(m^2)$ time (By Lemma~\ref{lem:bip:comp}), the following theorem holds:
\begin{theorem}
Maximal (edge-induced) Bipartite Subgraphs can be listed in $O(m^3)$ time delay.
\end{theorem}

\subsection{Complexity}
\label{sec:bitime}

In order to complete the analysis, let us look at the cost $\ccost$ for the three variants considered:

\begin{lemma}\label{lem:bip:comp}
$\ccost$ is $O(m)$ for Maximal Connected Induced Bipartite Subgraphs, $O(m+n\iack(n))$ for Maximal Induced Bipartite Subgraphs, and $O(m^2)$ for Maximal Edge-induced Bipartite Subgraphs,
where $n$ and $m$ are the number of vertices and edges, and $\iack(\cdot)$ is the functional inverse of the Ackermann function~\cite{Tarjan75UF}.
\end{lemma}
\begin{proof}
$\ccost$ is a bound for computing the $\comp(S)$ function as well as a canonical order. As the latter is computed by a BFS, it takes $O(m)$ time in all three cases, let us then focus on $\comp(S)$:

Firstly, observe that using $O(n)$ space and standard data structures, we can mark to which bipartition each vertex of $S$ belongs to (using $O(m)$ time to compute the initial bipartition of $S$), and which vertices have been already tested for addition, in $O(1)$ time per vertex. If a vertex fails the test to be added, it will not be possible to add it later on, so the total cost of $\comp(S)$ comes from selecting which vertices to test, and testing each of these vertices once.

For the connected case, we must test only vertices adjacent to $S$ (in no particular order): we can find these initial ``candidates'' in $O(\sum_{x\in S}|N(x)|) = O(m)$ time, marking each vertex as tested the first time so it is not tested again. Whenever trying to add a vertex $v$ to $S$, we must pay $O(|N(v)|)$ time to check that all its neighbors belong to the same bipartition of $S$, in which case $v$ belongs to the other one.
If $v$ is not addible, we immediately discard it. If instead we add it to $S$, we mark it with the correct bipartition, and update the list of candidate with its neighbors in $O(|N(v)|)$ time. As each vertex is only tested once and only added once, the total cost of $\ccost$ is bounded by $O(m)$. 

For the non-connected case, we further keep track of connected components via union-find~\cite{Tarjan75UF} (actually, for each connected component we will keep track of its two partitions). To test a vertex $v$ we must just check that it does not connect to two vertices in different partitions $C_0$ and $C_1$ of the same connected component $C$ of $X$: this can be done in $O(|N(v)|)$. Updating the union-find can be done in total $O(n \iack(n))$, where $\iack(\cdot)$ is the functional inverse of the Ackermann function~\cite{Tarjan75UF}.\footnote{As $\iack(n)$ grows extremely slowly, we remark that $\iack(n)$ is in essence $O(1)$ on real, finite, graphs.}

Once we tested a vertex, if this was not addible, it will never become addible, thus we only need to test each vertex once. The cumulative cost for testing will be the sum of the degrees of all tested vertices, that is bounded by $O(m)$.
The total time is thus $O(m + n \iack(n))$.

Finally, for Maximal Edge-induced Bipartite Subgraphs, we need to test each edge for addition just once as the property is hereditary. For each test we can simply check if the resulting graph is bipartite, which takes $O(m)$ time, for a total cost of $O(m^2)$.
\end{proof}


\section{Maximal k-Degenerate Subgraphs}\label{sec:list:first}
We here consider the enumeration of maximal $k$-degenerate subgraphs, giving an algorithm that has polynomial delay when $k$ is bounded.

A graph $G$ is $k$-degenerate if it allows an elimination order where each vertex has degree at most $k$ when deleted. 
Equivalently, it is $k$-degenerate if no subgraph of $G$ is a $(k+1)$-core, that is a graph where each vertex has degree greater or equal to $k+1$. 
The \textit{degeneracy} $d$ of $G$ is the smallest $k$ for which $G$ is $k$-degenerate.

A \textit{degeneracy ordering} of $G$ is an order of its vertices in which each vertex $v$ has at most $d$ neighbors occurring later than $v$, where $d$ is the degeneracy of $G$. It is well known that a degeneracy ordering can be found in $O(m)$ time by iteratively removing the vertex of smallest degree~\cite{DBLP:journals/corr/cs-DS-0310049}. To remove ambiguity, when multiple vertices have the same degree we can remove the one with smallest label.

The degeneracy is a well-known sparsity measure~\cite{EppsteinLS13}; its definition generalizes that of independent sets ($0$-degenerate graphs) and trees and forests (connected and non-connected $1$-degenerate graphs). Furthermore, degeneracy is linked to planarity as all planar graphs are $5$-degenerate, while outerplanar graphs are $2$-degenerate~\cite{lick1970k}.

We are interested in listing all maximal $k$-degenerate subgraphs of a graph $G$. An output-polynomial algorithm is known for maximal \textit{induced} $k$-degenerate subgraphs if $G$ is chordal~\cite{DBLP:conf/cocoon/ConteKOUW17}, but no output-polynomial results are known for general graphs.


\subsection{Maximal Induced k-Degenerate Subgraphs}\label{sec:kdegen:ind}

A subgraph of a $k$-degen\-erate graph is $k$-degenerate so the property is hereditary, and degeneracy can be computed in linear time so we can implement the $\comp(\cdot)$ function in polynomial time. 

Given a maximal induced $k$-degenerate subgraph $S$, we define its \textit{canonical order} as the \textit{reverse} of its degeneracy ordering, i.e., an ordering $s_1,\ldots, s_{|S|}$, such that $s_{|S|},\ldots, s_1$ is the degeneracy ordering of $S$. In the case of non connected subgraphs, this is adapted by considering the connected components one at a time in lexicographical order. Then, the proximity is defined by Definition~\ref{def:proximity}.

In the resulting ordering we have $|N(s_i)\cap \{s_1,\ldots,s_{i-1}\}| \le k$, i.e., the neighbors of $s_i$ in $S$ that precede $s_i$ in the canonical order are at most $k$. This is the key property that gives us the intuition for the algorithm: the removables correspond to all neighbors of the canonical extender \textit{except} a set of size at most $k$.
The neighboring function is obtained as follows.

\begin{definition}[Neighboring function for Maximal Induced $k$-Degenerate Subgraphs]\label{def:kdeg:neigh}
$$\neighs(S) = \bigcup\limits_{v\in V(G)} \neighs(S,v)$$
Where 
$\neighs(S,v) = \{\comp(\{v\} \cup S \setminus (N(v) \setminus K) : K\subseteq (S\cap N(v)) \text{ and } |K|\le k\}$
\end{definition}

Less formally, when computing $\neighs(S,v)$, we try to add $v$ to $S$ as canonical extender. Since $S$ is maximal, this violates the degeneracy constraint, so we remove all neighbors of $v$ except at most $k$ (the \textit{removable} set being $N(v)\setminus K$). The resulting subgraph $D = \{v\} \cup S \setminus (N(v) \setminus K)$ is $k$-degenerate: as $D\setminus \{v\}$ is $k$-degenerate as it is a subgraph of $S$, and any degeneracy ordering of $D\setminus\{v\}$ becomes a $k$-degenerate ordering for $D$ if we prepend $v$ in the beginning, because $v$ has at most $k$ neighbors in $D$. This means $N(v)\setminus K$ is a suitable removable according to Definition~\ref{def:crecon}.



We now show how these choices for $K$ satisfy Definition~\ref{def:crecon}: we iteratively try for $K$ all possible subsets of $S\cap N(v)$ of size at most $k$. These combinations, i.e., the number of removables, are $O(\sum_{i\in\{1, \ldots, k\}}\binom{|N(v)|}{i}) = O(n^k)$, which is polynomial when $k$ is bounded.

Let us now look at a target solution $S^*$ such that $v$ is the canonical extender for $S,S^*$, and let $s^*_1,\ldots,s^*_{|S^*|}$ be the canonical order of $S^*$: if $v = s_i$ in this order, it follows that $S\scap S^* = \{s^*_1,\ldots,s^*_{i-1}\}$, and $|N(v)\cap \{s^*_1,\ldots,s^*_{i-1}\}| \le k$. 

We also have that $\{s^*_1,\ldots,s^*_{i-1}\}\subseteq S$, so $N(v)\cap \{s^*_1,\ldots,s^*_{i-1}\} \subseteq N(v)\cap S$: since we try as $K$ all possible subsets of $N(v)\cap S$ of size at most $k$, we will eventually have $K =  N(v)\cap \{s^*_1,\ldots,s^*_{i-1}\}$.

At this point the neighboring function will yield $S' = \comp(\{v\} \cup S \setminus (N(v) \setminus K) = \comp(\{v\} \cup S \setminus (N(v) \setminus (S\scap S^*))$. In other words, we only remove some neighbors of $v$ from $S$, but all the neighbors that are part of $S\scap S^*$ are not removed, thus $S' \supseteq \{v\} \cup (S\scap S^*)$, meaning $|S'\scap S^*| > |S\scap S^*|$.





As for the running time, let us consider the cost $\ccost$ of a $\comp(X)$ call. $k$-degenerate graphs are hereditary, i.e., if a vertex is not addible it will not become addible later, so we need to test each $v\in V(G)\setminus X$ for addition at most once. As testing the degeneracy takes $O(m)$ time, $\ccost = O(mn)$ time.

Consider now $\neighs(S,v)$: firstly, we enumerate each possible $K \subseteq N(v)\cap S$, which takes $O(\sum_{i\in\{1,\ldots,k\}}\binom{|N(v)|}{i}) = O(|N(v)|^k)$ time. For each, we run $\comp(\{v\} \cup (S\setminus N(v)) \cup K)$, which takes $O(mn)$ time. The total cost is $O(n^{k+1} m)$ time.

The problem is thus proximity searchable, and the delay of the listing algorithm is the cost of $\neighs(S)$, i.e., running $O(n)$ times $\neighs(S,v)$ (maintaining $\sol$ is negligible). More formally:

\begin{theorem}\label{thm:kdegen}
Maximal Induced $k$-degenerate Subgraphs are proximity searchable when $k$ is constant, and can be enumerated in $O(mn^{k+2})$ time delay.
\end{theorem}

We now observe that $1$-degenerate subgraphs are exactly forests, and the connected ones are trees; setting $k=1$ we immediately obtain polynomial-delay algorithms for listing Maximal Induced Forests that could be easily adapted to Maximal Induced Trees. However, an ad-hoc analysis, delayed to Section~\ref{sec:indtrees}, shows we can obtain algorithms with better delay, and even reduce the space usage to polynomial for these problems.

\subsection{Maximal Edge-induced \textit{k}-Degenerate Subgraphs}

We now consider Maximal Edge-induced \textit{k}-Degenerate Subgraphs, i.e., maximal sets of edges $E\subseteq E(G)$ that correspond to a $k$-degenerate subgraph of $G$.
An algorithm for this case can be obtained by exploiting the structure of the induced one.
In the following, let $N_E(v)$ be the \textit{edge neighborhood} of $v$, i.e., the set of edges of $G$ incident to the vertex $v$.
Note that edge-induced $k$-degenerate subgraphs are also hereditary, and so $\comp(\cdot)$ takes polynomial time.

Let $S$ be an edge-induced $k$-degenerate subgraph, and let $v_1, \ldots , v_l$ be the canonical order of the vertices of $G[S]$ (i.e., the graph containing only edges in $S$ and vertices incident to them), as in Section~\ref{sec:kdegen:ind}.

The canonical ordering of $S$ is obtained by selecting the edges of $B$ by increasing order w.r.t. their \textit{later} endpoint in the vertex order, breaking ties by order of the other (earlier) endpoint.

This corresponds to selecting the vertices $v_1, \ldots , v_l$ in order, and for each adding the edges towards the preceding vertices one by one. Whenever all the edges from $v_i$ to the preceding vertices have been added, we can observe that the graph corresponds to that induced in $G[S]$ by the vertices $\{v_1, \ldots , v_i\}$. By the canonical order of the vertices defined in Section~\ref{sec:kdegen:ind}, this means $v_i$ has at most $k$ neighbors in $\{v_1, \ldots , v_{i-1}\}$.

Again, the proximity $\scap$ is given by Definition~\ref{def:proximity}.

We can now define the neighboring function:

\begin{definition}[Neighboring function for Maximal Edge-induced $k$-Degenerate\linebreak Subgraphs]\label{def:e:kdeg:neigh}
Let $S$ be a maximal edge-induced $k$-degenerate subgraph, and $e = \{a,b\}$ an edge not in $S$.
We define:
$$\neighs(S) = \bigcup\limits_{e=\{a,b\} \in E\setminus S} \neighs(S,a,b) \cup \neighs(S,b,a)$$
Where
$\neighs(S,a,b) = \{\comp(\{e\} \cup (S\setminus N_E(a)) \cup K) : K\subseteq (S\cap N_E(a)) \text{ and }$\linebreak $ |K|\le k-1\}$
\end{definition}

In other words, we add an edge $e=\{a,b\}$ to $S$, then force $a$ (or, respectively, $b$) to have degree at most $k$, by removing all other edges incident to it except at most $k-1$, as well as adding $e$.
The resulting graph is $k$-degenerate as $a$ (respectively $b$) has degree $k$, and the residual graph is a subgraph of $S$, which is $k$-degenerate, so it is possible to compute a degeneracy ordering.


Consider now two solutions $S$, $S^*$, with $S \scap S^* = \{e_1, \ldots , e_h\}$, and let $\de = \{x,y\}$ be the earliest edge in the canonical order of $S^*$ that is not in $S$, i.e, $e_{h+1}$.
Assume wlog that $x$ comes before $y$ in the canonical (vertex) ordering of $S^*$. In this ordering, $y$ has at most $k$ neighbors preceding it, i.e, $|\{e_1, \ldots , e_h\} \cap N_E(y)| \le k$. Furthermore, by the same definition, all edges incident to $y$ that precede $\de$ in the ordering must be between $y$ and another vertex which comes earlier than $x$, and thus than $y$, in the ordering, thus they may be at most $k-1$ ($k$, including $\de$ itself, from $y$ to $x$). Let $K'$ be the set of these edges (not including $\de$).

When computing $\neighs(S,y,x)$, we consider all subsets of edges in $S$ incident to $y$ of size at most $k-1$. By what stated above, at some point we will consider exactly $K'$. In this case, we will obtain $S' = \comp(\{\de\} \cup (S\setminus N_E(y)) \cup K')$. This must contain all edges in $\{e_1, \ldots , e_h\}$, as we only removed edges neighboring $y$, but all those in $\{e_1, \ldots , e_h\}$ were in $K'$. Thus we have $\{e_1, \ldots , e_h\} \cup \de = \{e_1, \ldots , e_h, e_{h+1}\} \subseteq S'$, which implies $|S'\scap S^*| > |S\scap S^*|$.
The case in which $x$ comes after $y$ in the ordering is similarly satisfied by $\neighs(S,x,y)$.

Finally, we only need to show that $\neighs(S)$ takes polynomial time to compute: indeed this is $O(m)$ times the cost of $\neighs(S,y,x)$, which in turn has the cost of computing $\comp(\cdot)$ once for each possible considered set $K$. These latter are $O(\binom{N_E(y)}{k-1})$, and the $\comp(\cdot)$ can be easily implemented in $O(m^2)$ (as above, testing degeneracy takes $O(m)$ time and each edge needs to be considered at most once for addition since the problem is hereditary), for a total cost that is polynomial when $k$ is constant. We can thus state the following:

\begin{theorem}\label{thm:e:kdegen}
Maximal Edge-induced $k$-degenerate Subgraphs are proximity searchable when $k$ is constant, and can be enumerated with delay $O(\binom{n}{k-1}m^3)$.
\end{theorem}

\section{Maximal Chordal Subgraphs}\label{sec:chordal}

\subsection{Maximal Induced Chordal Subgraphs}
A graph $G$ is chordal if every cycle in $G$ of length greater than 3 has a chord, i.e., an edge between two non-consecutive vertices in the cycle. 
Chordal graphs have been widely studied, and it is known that several problems which are challenging on general graphs become easier on chordal graphs (see, e.g.,~\cite{Chandran2001,Okamoto2005,blair1993introduction}). While the problem of finding a largest chordal subgraph has been studied~\cite{bliznets2016largest}, to the best of our knowledge there are no known enumeration results.

We here aim at listing Maximal Induced Chordal Subgraphs of $G$.
The problem is hereditary, and chordality can be tested in $O(m)$ time~\cite{rose1976algorithmic}, thus $\comp(\cdot)$ takes $O(mn)$ time.

A (sub)graph is chordal iff it allows a \textit{perfect elimination ordering} $\{v_1, \ldots, v_n\}$ of its vertices, i.e., such that for all $i$, $N(v_i) \cap \{v_{i+1},\ldots,v_{n}\}$ is a clique~\cite{DBLP:conf/cocoon/ConteKOUW17}.
We can obtain this by recursively removing simplicial vertices, i.e., vertices whose neighborhood in the residual graph is a clique.\footnote{To remove ambiguity, we can remove the lexicographically smallest when multiple simplicial vertices are present.} 

As the neighbors of a simplicial vertex form a clique, we observe that removing a simplicial vertex cannot disconnect the residual graph.
It is also known that a chordal graph has $O(n)$ maximal cliques, and a vertex $v$ participates in $O(|N(v)|)$ maximal cliques~\cite{DBLP:conf/cocoon/ConteKOUW17}.



We use this to define the canonical order, which is then combined with Definition~\ref{def:proximity} to obtain the proximity function~$\scap$.

\begin{definition}[Canonical Order for Maximal (Connected) Induced Chordal Subgraphs]\label{canon:chordal}

The canonical order $\{s_1, \ldots, s_{|S|}\}$ of $S$ is the \textit{reverse} of its perfect elimination ordering, i.e., such that $\{s_{|S|}, \ldots, s_{1}\}$ is the perfect elimination ordering.
\end{definition}

This way, the neighbors of $v$ that precede $v$ in the ordering form a clique.  
Furthermore, when $S$ is a connected subgraph, any prefix $\{s_1, \ldots, s_{j\le |S|}\}$ of the canonical order induces a connected subgraph, because we can iteratively remove the last vertex, which is always simplicial. This means the canonical order satisfies condition~(\ref{item:canonical1}) of Definition~\ref{def:crecon}, in the case of both Maximal Induced Chordal Subgraphs and Maximal Connected Induced Chordal Subgraphs.
The neighboring function is defined as follows.

\begin{definition}[Neighboring function for Maximal (Connected) Induced\linebreak Chordal Subgraphs]\label{neighboring:chordal}~

We define $\neighs(S) = \bigcup\limits_{v\in V(G)\setminus S} \neighs(S,v)$.

\noindent For the non connected case we define 

$\neighs(S,v)= \{\comp(S \cup \{v\} \setminus ( N(v) \setminus Q )) :$

$Q \text{ is a maximal clique of } G[S \cup \{v\}] \text{ containing v} \}$

\noindent While for the connected case we define 

$\neighs(S,v)=\{\comp( \ccomp_v(S \cup \{v\} \setminus ( N(v) \setminus Q )) ) :$ 

$Q \text{ is a maximal clique of } G[S \cup \{v\}] \text{ containing v} \}$


\end{definition}

Less formally, we add a vertex $v$ to $S$, then remove all its neighbors except one maximal clique $Q$ (meaning the \textit{removable} by definition of canonical reconstruction will be $N(v)\setminus Q$). In the connected case, we further remove vertices not in the connected component of $v$. 

We can easily see that $S \cup \{v\} \setminus ( N(v) \setminus Q )$ is chordal, by showing a perfect elimination ordering: $v$ itself is simplicial as its neighbors form a clique, and can be removed; we can then complete the perfect elimination order as the remaining vertices form an induced subgraph of $S$, which is chordal as induced chordal subgraphs are hereditary.

We now need to prove the last condition; let $S$ and $S^*$ be two solutions, $S\scap S^* = \{s^*_1, \ldots, s^*_h\}$ and $\dv = s^*_{h+1}$ the earliest vertex in the canonical order of $S^*$ not in $S$.

By the canonical order, $N(\dv) \cap (S\scap S^*) = N(\dv) \cap \{s^*_1, \ldots, s^*_h\}$ is a clique. When computing $\neighs(S,\dv)$, as we try all maximal cliques, for some $Q$ we will have $N(\dv) \cap (S\scap S^*) \subseteq Q$. The resulting $S'$ will thus contain all neighbors of $\dv$ in $S\scap S^*$, and thus all of $S\scap S^*$, plus \dv, meaning that $|S'\scap S^*|>|S\scap S^*|$, which proves the correctness the $\neighs(\cdot)$ function.

Finally, $\neighs(\cdot)$ can indeed be computed in polynomial time.
We first need to list all maximal cliques containing $v$ in $G[S \cup \{v\}]$: these correspond exactly to the maximal cliques of $G[(S\cap N(v))\cup \{v\}]$; as $v$ is adjacent to \textit{all} vertices in $(S\cap N(v))$, we can further say that these correspond exactly to all maximal cliques of $G[(S\cap N(v))]$, to which we then add $v$ ($v$ can clearly be added to any clique of $G[(S\cap N(v))]$ since it is adjacent to all its vertices).

This correspondence is important, because $G[(S\cap N(v))]$ is an induced subgraph of $S$, and thus a chordal graph.


Recall now that a chordal graph has $O(n)$ cliques and they can be listed in $O(m)$ time (e.g., by computing a perfect elimination ordering~\cite{rose1976algorithmic}): $G[S\cap N(v)]$ has at most $|N(v)|$ vertices and $O(|N(v)|^2)$ edges, so we can list all the maximal cliques of $G[S\cap N(v)]$ --and thus all maximal cliques of $G[S \cup \{v\}]$ containing $v$-- in $O(|N(v)|^2)$ time, obtaining at most $|N(v)|$ maximal cliques.

The cumulative cost of listing all cliques for each $v\in V(G)\setminus S$ is thus bounded by $O(\sum\limits_{v\in V(G)\setminus S}|N(v)|^2) = O(mn)$ time, and the process yields $O(\sum\limits_{v\in V(G)\setminus S}|N(v)|) = O(m)$ maximal cliques. For each clique $Q$, we must further compute the corresponding $\comp(\cdot)$ call: as the problem is hereditary, again we only need to test each vertex at most once for addition, and a chordality can be tested in $O(m)$ time, the cost $\ccost$ of a $\comp(\cdot)$ is $O(mn)$ time (which dominates the time for checking membership in $\sol$). Furthermore, the same bound applied to the connected case, as we simply need to consider vertices for addition only when they become adjacent to the current solution. Scanning the neighborhoods of the vertices that are added to the solution to find these candidates has an additional cost of $O(m)$ which does not affect the $O(mn)$ bound. The total cost will be $O(mn + m\cdot mn) = O(m^2n)$

We can thus state that:

\begin{theorem}
Maximal Induced Chordal Subgraphs and Maximal Connected Induced Chordal Subgraphs are proximity searchable, and can be listed with $O(m^2n)$ time delay.
\end{theorem}

\subsection{Maximal Edge-induced Chordal Subgraphs}

An algorithm for the edge version can be obtained by defining the canonical order for the edge-induced subgraph in the same way as for Bipartite Subgraphs, based on the canonical ordering of the vertices (see Definition~\ref{def:bip-non-canon}). 

In this problem too, note how all Maximal Edge-induced Chordal Subgraphs of a connected graph are connected, as we can always add edges to a non-connected subgraph without creating cycles, so we do not need to separately consider the connected and non-connected case.

We can then devise a neighboring function $\neighs(S,(x,y))$ like the first one in Definition~\ref{neighboring:chordal}, where we use an edge $(x,y)$ as canonical extender.

When adding an edge $(x,y)$ to a maximal solution $S$, we try as $Q$ all maximal cliques containing either $x$ or $y$ in $G[S]$.
In any $S^*$ (for which $(x,y)$ is the canonical extender) one between $x$ and $y$ will occur later in the canonical ordering; wlog, let us say $y$. 
As the canonical ordering is based on a reversed perfect elimination ordering, the neighbors of $y$ preceding $y$ in the canonical order of $S^*$ form a clique (including $x$ as well). Thus the neighboring funciton will eventually consider a clique $Q$ containing $y$ and all its preceding neighbors, and when this happens the proximity with $S^*$ is extended.
The number of neighboring solutions generated this way will be $O(\sum_{(x,y)\in E(G)}|N(x)|+N|(y)|) = O(mn)$.

The only further requirement is a polynomial time $\comp(\cdot)$ function which needs to be applied to each neighboring solution: this follows from~\cite{HEGGERNES20091}, who prove that edge-induced chordal subgraphs are \textit{sandwich monotone}. In other words, if a edge-induced chordal subgraph $S\subseteq E(G)$ is \textit{not} maximal, then there is always a single edge $e\in E(G) \setminus S$ such that $S\cup \{e\}$ is a chordal subgraph.

This means $\comp(\cdot)$ can be computed in a greedy way by testing, up to $m$ times, that any of the $O(m)$ remaining edges in the graph can be added, which takes $O(m)$ time, for a total cost $\ccost = O(m^3)$. The total cost follows.

\begin{theorem}
Maximal Edge-induced Chordal Subgraphs are proximity searchable, and can be listed with $O(m^4n)$ time delay.
\end{theorem}

\section{Maximal Induced Proper Interval Subgraphs}\label{sec:interval-exp}
Interval graphs are a well-known subclass of chordal graphs, whose vertices can be arranged as intervals on a line such that two vertices are adjacent if and only if their intervals intersect. In this section, we present a polynomial-delay enumeration algorithm for Maximal \textit{Proper} Interval Subgraphs, a subclass of interval subgraphs corresponding to interval graphs where no two intervals properly contain another.

Despite chordal graphs, interval graphs, and proper interval graphs being closely related to each other, it is interesting to observe how the three enumeration algorithms proposed here (chordal subgraphs, proper interval subgraphs) and in~\cite{Cao:arXiv:2020} (interval subgraphs) differ significantly. 
Furthermore, an interesting open question would be to determine whether it is possible to enumerate Maximal Interval Subgraphs directly via proximity search, or whether there is an intrinsic difference in what can be achieved with retaliation-free paths.

\subsection{Maximal Connected Induced Proper Interval Subgraphs}

A \textit{proper} interval graph is an interval graph where, in the interval representation, no interval properly contains another. Equivalently, it can be defined as interval graphs that admit a \textit{unit-length} representation, i.e., where all intervals have length $1$~\cite{fulkerson1965incidence}. In this section we will adopt this latter definition, and all interval representations considered will be intended as unit-length.

In the following, we show how to enumerate Maximal (Connected) Proper Interval Subgraphs of a graph $G$. We show the connected version of the problem, and later remark how to adapt it to the non-connected case.


It is important to observe that every connected proper interval (sub)graph $S$ has two unique interval representation represented by a sequence $v_1, \ldots, v_{|S|}$, and its reverse.\footnote{Ambiguity may be caused by identical vertices, i.e., adjacent and with the same sets of neighbors, but it can be resolved by taking the smallest-id vertex first.} The \textit{canonical order} $v_1, \ldots, v_{|S|}$ of a Maximal Proper Interval Subgraph $S$ is defined as the sequence given by the interval representation of $S$ which has as $s_1$ the smallest among the two possible values.

A graphical example is given in Figure~\ref{fig:pinterval} (a),(b),(d),(e).

\begin{figure*}
    \centering
    \includegraphics[width=\textwidth]{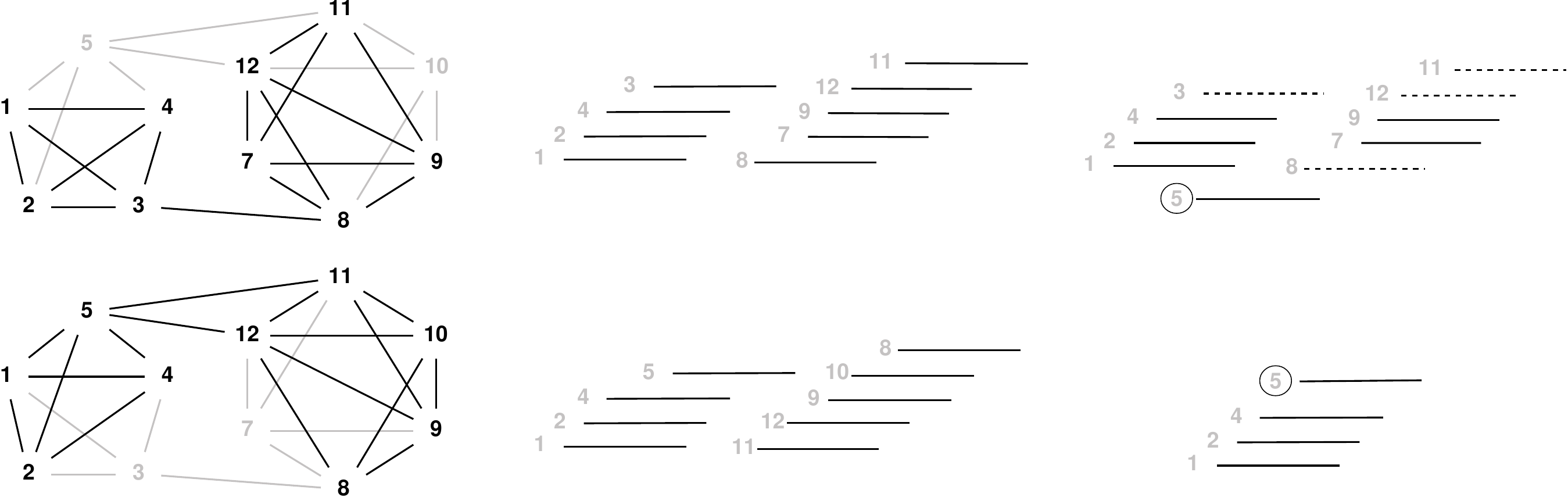}
    \put (-505,145){\small(a)}
    \put (-320,145){\small(b)}
    \put (-140,145){\small(c)}
    \put (-505,60){\small(d)}
    \put (-320,60){\small(e)}
    \put (-140,60){\small(f)}
    %
    %
\caption{\textbf{a}: a proper interval subgraph $S$ of the graph in Figure~\ref{fig:ex:bip}. \textbf{d}: another proper interval subgraph $T$. \textbf{b}: the unit interval representation of $S$, which induces the canonical order $1,2,4,3,8,7,9,12,11$. \textbf{e}: the representation of $T$, which induces the canonical order $1,2,4,5,11,12,9,10,8$.}
    \label{fig:pinterval}
\end{figure*}
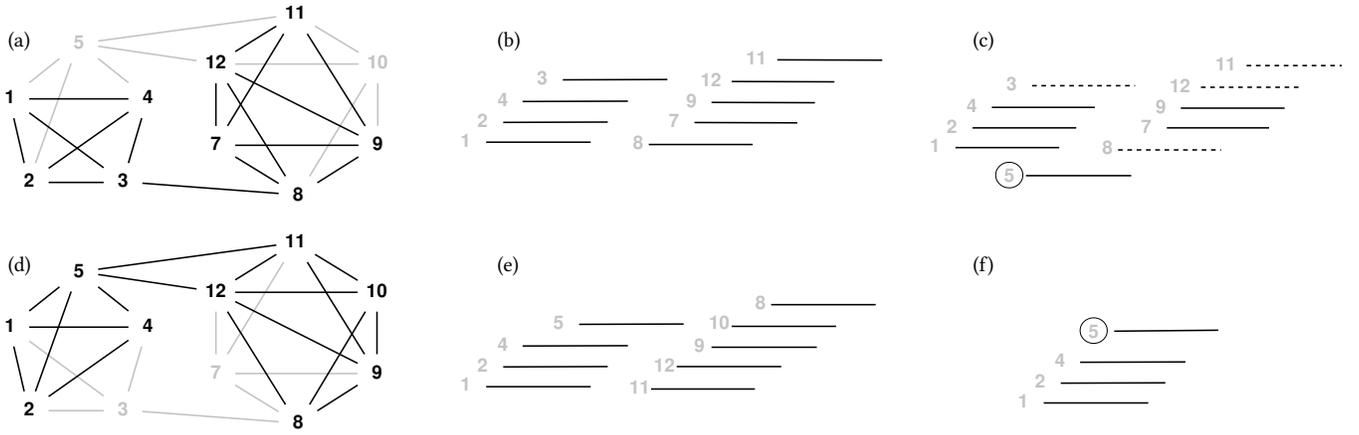

In order to generate neighboring solutions, we take each of the two representations, and for each we proceed as follows. Firstly, we try all vertices $v\in V(G)\setminus S$ as canonical extender.
Next, we want to identify the right way of inserting $v$ in the interval representation of $S$. 
As the interval length is fixed to 1, there are $5|S|$ possibilities: for any other $w \in S$, we can place $v$ ending just before / just after the start of $w$, exactly overlapping $w$, or starting just before / just after the end of $w$. We can observe how these capture all possible ways to insert $v$, since given any other placement we can slide $v$ in any direction until it is about to gain or lose an overlap, i.e., one of the endpoints of $v$ will approach the endpoint of another interval, which puts us into one of the $5|S|$ cases above.

Finally, we have to make the representation consistent, by removing from $S\cup \{v\}$:
\begin{itemize}
\item All vertices (intervals) coming \emph{after} $v$ in this representation.
\item All neighbors of $v$ that do \emph{not} overlap its interval.
\item All non-neighbors of $v$ that do overlap its interval.
\item Finally, all vertices that have become disconnected from the connected component containing $v$ by performing the previous steps.
\end{itemize}

The resulting representation is consistent with a proper interval graph, so the resulting graph is clearly a proper interval graph. Let us call it $X$, and let us call $S' = \comp(\ccomp_{v}(X))$.
The following example illustrates these operations on the example graphs from Figure~\ref{fig:pinterval}.

\begin{example*}
Looking at the graphs in Figure~\ref{fig:pinterval} (\textbf{(a)}, \textbf{(d)}) and their canonical orderings (see \textbf{(b)},\textbf{(e)}) we can give an example of the operations performed by the neighboring function.

\textbf{(c)}: $5$ is added as canonical extender and placed ``just before'' $3$; we remove vertices the dashed vertices $3$ and $8$ (overlapping $5$ but not neighbors of $5$ in $G$), and $11$ and $12$ (not overlapping $5$ but neighbors of $5$ in $G$); $9$ and $7$ are also removed as not part of the same connected component as $5$. The resulting graph \textbf{(f)} is then is maximalized with $\comp(\cdot)$, and has greater proximity with $b$ than $a$.

Indeed, $S\scap T = \{1,2,3,4\}$, while the maximal solution obtained maximalizing \textbf{(f)} contains $1,2,3,4,5$.
\end{example*}

The set $\neighs(S)$ will be made of the $S'$ obtained by trying both representations of $S$, for each all $v\in V(G)\setminus S$, and for each all possible insertions, for a total of $O(n^2)$ neighboring solutions.

We now show that, given $S,T$, we always obtain some $S'\in \neighs(S)$ such that $|S'\scap T| > |S\scap T|$.

Let the proximity $S\scap T$ be $t_1, \ldots, t_{i-1}$: this is a prefix of $T$, and since it is a connected subgraph, its vertices in the same order as in $T$ in one of the two representations of $S$. Furthermore, let $v$ be our canonical extender.

Looking at the canonical order of $T$, consider the placement of $t_i$ relative to the preceding intervals $t_1, \ldots, t_{i-1}$. Now, consider the case where $v = t_i$, the correct one between the two representations of $S$ is considered, and the placement of $v$ relative to $t_1, \ldots, t_{i-1}$ is the same as $t_i$ in $T$.

As we try all possible placements for $v = t_i$, and as $t_1, \ldots, t_{i-1}$ is a connected subgraph of $T$, 
we will also try the placement of $t_i$ considered above.
The correct placement of $v$ tells us that, when we remove all intervals coming \emph{after} $v$ from $S\cup \{v\}$, we do not remove vertices from $t_1, \ldots, t_{i-1}$. The placement also tells us that all neighbors of $v$ in $t_1, \ldots, t_{i-1}$ overlap with $v$ in the interval representation, while all non-neighbors of $v$ in $t_1, \ldots, t_{i-1}$ do not, so vertices of $t_1, \ldots, t_{i-1}$ are not removed in the remaining steps.
It follows that the set $X$ obtained is a proper interval subgraph of $G$ containing $t_1, \ldots, t_{i}$, and that $S' = \comp(X)$ is such that $|S'\scap T| > |S\scap T|$.

\subsection{Induced Proper Interval Subgraphs}

The non-connected version is similarly solved: as in previous sections, we define the canonical order by ordering each connected component as in the connected case, then ordering the components by smallest-id vertex.

The proximity is then defined by canonical reconstruction (Definition~\ref{def:proximity}).

Let $S_1,\ldots, S_j$ be the connected components of $S$, and $T_1, \ldots, T_k$ those of $T$, for some $T$.

We will have that $S\scap T$ consists of some complete connected components of $T$, as well as a (possibly empty) subset of one component $T_i$. The canonical extender $t_i$ will be the earliest vertex of the canonical ordering of $T_i$ that is not in $S$.
We can also observe that each of these components is contained itself in some connected component of $S$.

When we select a canonical extender $v$, we also select a component $S_i$ (of course, trying all possibilities): we aim to find the component which contains the partially formed $T_i$.

As we add $v$ to $S$, we can immediately remove from $S$ all neighbors of $v$ in $S\setminus S_i$, as indeed as $t_i$ only has neighbors in $T_i$.

We then proceed as in the connected case, ordering $S_i$ in the two possible ways, and trying all $O(n)$ insert possibilities. 
However, we must take care of the fact that $S_i$ may contain more than one connected component of $T$: components \textit{preceding} $v$ in the order are preserved (they cannot overlap $v$ as they precede the other elements of $T_i$, to which they are not adjacent), but there may be some \textit{following} $v$. 

For this reason, we must introduce another ``guess'', that is we guess which of the intervals preceding $v$ is $t_{i-1}$ (as, to obtain the correct placement of $v$ in the interval representation, we may not have placed it just after $t_{i-1}$). Note how the number of possible $t_{i-1}$ is at most $|N(v)|$.
This step, which was not necessary in the connected version, allows us to effectively compute a safe way of detaching the intervals following $v$.


Specifically, we remove from $S_i \cup \{v\}$ the following vertices:

\begin{itemize}
\item All vertices (intervals) between $t_{i-1}$ and $v$.
\item All vertices (intervals) coming \emph{after} $v$ in this representation that are adjacent to $v$ or to $t_{i-1}$.
\item All neighbors of $v$ that do \emph{not} overlap its interval.
\item All non-neighbors of $v$ that do overlap its interval.
\end{itemize}

In this way, all the vertices removed could not be part of the proximity $S\scap T$: the ones between $t_{i-1}$ and $t_i$ cannot be in the interval representation since $t_{i-1}$ and $t_i$ are consecutive in the interval representation of $T$; the others we remove were neighbors of $t_{i-1}$ or $t_i$, but did not precede them in the interval representation (i.e., in the canonical order), so they could not be part of $S\scap T$. On the other hand, the remaining intervals following $v$ are now in a separate connected component, thus the interval representation of the resulting graph is consistent, contains $t_1, \ldots, t_i$ (for any $T$, when the correct choices are performed), and we can apply the $\comp(\cdot)$ function to obtain a neighboring solution $S'$.

\subsection{Running time}

As for the complexity, the neighboring function tries $O(n)$ candidates for $v$, and for each, $5|S| = O(n)$ possible placements in $S$, in each of the $2$ representations.
For each, the cost of the procedure is dominated by the application of the $\comp(X)$ call.

Since we can test whether a graph is a proper interval graph in $O(m)$ time~\cite{booth1975linear}, the $\comp(X)$ function can be implemented in $O(mn)$ time as for chordal subgraphs, giving us a total cost per solution of $O(n^3m)$ time.

For the non-connected version, we must consider two additional phases: for each $v$, we firstly selected a connected component $S_i$ of $S$, and secondly we selected the possible $t_{i-1}$ among the neighbors of $v$. We can bound the number of connected components by $O(n)$, and, rather than adding another factor $O(n)$ for the choice of $t_{i-1}$, we can observe how each distinct pair $v,t_{i-1}$ corresponds to an edge, so the number of possible $v$ and $t_{i-1}$ pairs is $O(m)$, for a total complexity of $O(n^3m^2)$

\begin{theorem}
Maximal Induced Proper Interval Subgraphs and Maximal Connected Induced Proper Interval Subgraphs are proximity searchable, and can be listed with $O(n^3m^2)$ time delay and $O(n^3m)$ time delay, respectively.
\end{theorem}

\section{Maximal Obstacle-free Convex Hulls}
\newcommand{\shad}{\phi\xspace}
\newcommand{\moc}{\textsc{moc}\xspace}
\newcommand{\mocs}{\textsc{moc}s\xspace}

In application domains such as robotics planning and routing, a common problem is finding areas, typically convex, in a given environment which are free from obstacles (see, e.g.,~\cite{Deits2015,savin2017algorithm}).

In this section we solve the following formulation of the problem: let $V$ and $X$ be two sets of elements, which corresponds to points on a 2-dimensional plane. $V$ represents the point of interest for our application, and $X$ represents the obstacles. For short, let $|V|=j$ and $|X|=h$, and let $n = j+h$ be the total number of points. We are interested in listing all maximal obstacle-free convex hulls (\mocs for short), where an obstacle-free convex hull is a set of elements $S\subseteq V$ such that the convex hull of $S$ does not contain any element of $X$.

This problem does not concern a graph, but its solutions are modeled as sets of elements, thus the technique may still be applied. Furthermore, we can naturally generalize the problem by adding a graph structure to $V$, i.e., adding edges between its points, and considering the problem of Maximal \textit{Connected} Obstacle-free Convex hulls.

\subsection{Maximal Obstacle-free Convex Hulls}

Again, note that the problem is hereditary, i.e., each subset $S'$ of a solution $S$ clearly also admits a convex hull which does not include elements of $X$ (since it will be contained in that of $S$).

It is worth observing that this is the only problem in this paper to which we do \textit{not} apply the canonical reconstruction strategy.

Consider a maximal solution $S$ and an element $v\in V\setminus S$. As $S$ is maximal, there is at least one element $x\in X$ included in the convex hull of $S\cup \{v\}$. This element $x$ casts two ``shadows'' $S_1$ and $S_2$ on $S$, seen by $v$: consider the straight line between $v$ and $x$, $S_1$ consists of all elements of $S$ above this line, and $S_2$ of all those below it. It is straightforward to see how both the convex hull of $S_1\cup \{v\}$ and that of $S_2\cup \{v\}$ do not contain $x$. Any element of $S$ that falls exactly on the line may not participate in any solution involving $v$.\footnote{Note that it may not fall between $v$ and $x$ otherwise the convex hull of $S$ would have included $x$.} Furthermore, any element $x'\in X$ above this line, and still in the convex hull of $S\cup \{v\}$, further casts two shadows on $S_1$, as any element below this line casts them on $S_2$. If we repeat this process for all elements of $X$ in the convex hull of $S\cup \{v\}$ we obtain a number of shadows of $S$ which is at most linear in the number of elements of $X$. Let $\shad(S,v)$ be the set of these shadows.
For each of these shadows $S_i\in \shad(S,v)$, we have that the convex hull of $S_i\cup \{v\}$ may not include elements of $X$, i.e., $S_i\cup \{v\}$ is a (possibly not maximal) solution.

The neighboring function is then obtained as follows.

\begin{definition}[Neighboring function for \mocs]~

$$\neighs(S) = \bigcup\limits_{v\in V(G)\setminus S} \neighs(S,v)$$

Where 
$$\neighs(S,v) = \{\comp(S_i \cup \{v\}) : S_i \in \shad(S,v)\}$$
\end{definition}

Finally, for two solutions $S$ and $S^*$, we simply define $S \scap S^*$ as the intersection $S\cap S^*$ between their elements.

Let $I = S\cap S^* = S\scap S^*$, and $v$ any element in $S^*\setminus S$. Since $I\cup \{v\}$ is contained in a \moc, $S^*$, its convex hull cannot contain any element of $X$. It follows that $I$ must be fully contained in a single $S_i \in \shad(S,v)$: indeed, if we take two points $u_i\in S_i$ and $u_j\in S_j$, it is evident by the definition of $\shad(S,v)$ that the convex hull of $\{v,u_i,u_j\}$ (or any superset of it) contains at least an element of $X$. 
We have that the neighboring function will return $S'= \comp(S_i\cup \{v\})$, with $I\cup\{v\} \subseteq S'$, which implies $|S'\scap S^*| > |S\scap S^*|$. The algorithm is thus correct.

As for the complexity, the problem is hereditary, so we may compute a $\comp(S)$ call by testing each vertex in $V\setminus S$ once. The convex hull of $S$ can be computed in $O(|S|\log |S|)$ time~\cite{Chan1996}, and testing a solution consists in checking that each vertex of $X$ is not in this hull, which can trivially be done in $O(|S|\cdot h)$ time. The cost of $\comp(\cdot)$ is thus $O(j(h+\log j))$ time.
For each candidate $v$, we have at most $h$ neighboring solutions, and since we need to consider at most $j$ candidates, the delay of the algorithm will be $j\cdot h$ times the cost of a $\comp(\cdot)$ call.

We thus obtain an algorithm with the following complexity:

\begin{theorem}
Maximal Obstacle-free Convex Hulls are proximity searchable, and can be listed in $O(j^2h(h+\log j)) = O(n^4)$ time delay.
\end{theorem}

It could be argued that the neighboring function actually reports all solutions of the input-restricted problem in this instance, allowing us to induce a parent-child relationship with the structure of~\cite{lawler1980generating,Cohen20081147}, and reducing the space usage to $O(n)$ by using stateless iteration~\cite{DBLP:conf/icalp/ConteGMV16}. However, it is worth observing that proximity search required proving a weaker statement, and allows for an arguably simpler proof.


\subsection{Maximal Connected Obstacle-free Convex Hulls}

We now consider an extension of the problem where on top of $V$ and $X$ we have a graph structure $G = (V,E)$ on the points of $V$, and we are interested in listing all maximal set of points $S\subseteq V$ such that the convex hull of $S$ is obstacle-free, and $G[S]$ is connected.

We consider this a natural extension as, in the applications mentioned above, it could model requirements on the structure of the obstacle-free areas identified.

The algorithm is remarkably similar to the above version, as the neighboring function still considers $S_i \cup \{v\}$ for all $S_i\in \shad(S,v)$, but only keeps the connected component of $G[S_i \cup \{v\}]$ containing $v$.

\begin{definition}[Neighboring function for \mocs]~

$$\neighs(S) = \bigcup\limits_{v\in V(G)\setminus S} \neighs(S,v)$$

Where 
$$\neighs(S,v) = \{ \ccomp_v(\comp(S_i \cup \{v\}) ) : S_i \in \shad(S,v)\}$$
\end{definition}

For two solutions $S$ and $S^*$, we define $S \scap S^*$ as the largest connected component of their intersection $S\cap S^*$. We now prove that there is $S'\in \neighs(S)$ such that $|S'\scap S^*| > |S\scap S^*|$.

Let $I = S\cap S^* = S\scap S^*$, and $v$ any element in $S^*\setminus S$ such that $G[I\cup \{v\}]$ is connected.
Note that a suitable $v$ must exist, otherwise $I$ would not be connected to the elements of $S^*\setminus I$, contradicting the fact that $S^*$ is a connected solution.
Since $I\cup \{v\}$ is contained in $S^*$, its convex hull cannot contain any element of $X$. It follows that $I$ must be fully contained in a single $S_i \in \shad(S,v)$.
Furthermore, as $I\cup \{v\}$ is connected, it must be contained in $\ccomp_v(S_i \cup\{v\})$, the connected component of $G[S_i \cup\{v\}]$ containing $v$.

Similarly to the above case, we have that the neighboring function will return $S'= \comp(\ccomp_v(S_i\cup \{v\}))$, with $I\cup\{v\} \subseteq S'$, which implies $|S'\scap S^*| > |S\scap S^*|$. The algorithm is thus correct.

The complexity can be also derived from the non-connected case: the only additional step is applying $\ccomp_v(\cdot)$ before the $\comp(\cdot)$ function. As $\ccomp_v(\cdot)$ takes $O(m)$ time, where $m = |E(G)| = O(n^2)$, we can conclude the following:

\begin{theorem}
Maximal Connected Obstacle-free Convex Hulls are proximity\linebreak searchable, and can be listed in $O(jh (m + j(h+\log j))) = O(n^4)$ time delay.
\end{theorem}

\section{Maximal Connected Directed Acyclic Subgraphs}\label{sec:list:last}
\newcommand{\mcais}{\textsc{mcais}\xspace}

In this section we consider a \textit{directed} graph, where each edge has a head and a tail, and its direction is from the tail to the head. We call $N^+(v)$ the \textit{out-neighbors} of the vertex $v$ and $N^-(v)$ its \textit{in-neighbors}.

The goal of this section is listing Maximal Induced Connected Acyclic  Subgraphs (\mcais hereafter) of a given directed graph $G$.
The problem is connected-hereditary, 
and acyclicity can be tested in $O(m)$ time, thus $\comp(\cdot)$ can be implemented in $O(mn)$ time.

For completeness, we remark that the non-connected version (Maximal Induced Directed Acyclic Subgraphs), corresponds to listing the complements of Minimal Feedback Vertex Sets in a directed graph, and is of no interest here as an output-polynomial algorithm is given in~\cite{schwikowski2002enumerating}. We thus address the connected version of the problem, which has no natural counterpart in terms of feedback vertex set.
Let us define the canonical order:

\begin{definition}[Canonical Order for Maximal Connected Induced Acyclic Subgraphs]\label{def:mcaisorder}
The canonical order of a \mcais $S$ is the order $\{s_1, \ldots, s_{|S|}\}$ such that, for each $s_i$, $\{s_1, \ldots, s_{i}\}$ is connected, and either $\{s_1, \ldots, s_{i-1}\}\cap N^+(s_i) = \emptyset$ or $\{s_1, \ldots, s_{i-1}\}\cap N^-(s_i) = \emptyset$.
If multiple orders are possible let it be the lexicographically minimum.
\end{definition}

Our algorithm does not need to compute this order or $\scap$, but we need to show that it always exists.


Firstly, recall that every acyclic graph has at least one source and one target, and let us observe an important property of acyclic graphs with a single source (whose proof trivially follows from the fact that any non-source vertex has a neighbor occurring before itself in the order):

\begin{lemma}\label{lem:topcon}
Let $G$ be a single-source acyclic connected graph, and $v_1, \ldots, v_n$ any topological order of $G$. Any prefix $v_1, \ldots, v_i$ of this order induces a connected subgraph.
\end{lemma}

Lemma~\ref{lem:topcon} also implies that the \textit{reversed} topological order (i.e., where vertices have no forward out-neighbors) of a single-target acyclic connected graph is such that every prefix induces a connected subgraph. We also remark that both these orders satisfy the intersection properties of Definition~\ref{def:mcaisorder}.

We now use this lemma to show that the defined canonical order exists for any \mcais.
In the following, we define \textit{collapsing} a set of vertices $A \subseteq S$ into $x$ as replacing them with a single vertex $x$, whose in- and out-neighbors correspond to all vertices in $S\setminus A$ that were respectively in- and out-neighbors of some vertex in $A$.

\begin{lemma}
Every Directed Acyclic Graph allows a canonical order by Definition~\ref{def:mcaisorder}.
\end{lemma}
\begin{proof}
Let $S$ be a Directed Acyclic Graph. Let $v_1$ be a source of $S$, and $S_1$ be the set of vertices reachable by $v_1$, including $v_1$. 
Let $s_{1,1}, \ldots, s_{1,|S_1|}$ a topological ordering of $S_1$. 


No vertex in $S_1$ can have an out-neighbor outside of $S_1$ as otherwise said vertex would be in $S_1$ itself. Let instead $S_2$ be the set of all vertices in $S\setminus S_1$ that can reach some vertex of $S_1$. 

If we collapse $S_1$ into a vertex $x$, we can observe that $S_2 \cup \{x\}$ is acyclic subgraph with $x$ being the only target. Let $x, s_{2,1}, \ldots, s_{2,|S_2|}$ be a reverse topological ordering of $S_2 \cup \{x\}$. 

If we replace $x$ with the previously computed order of $S_1$, we obtain an order $s_{1,1}, \ldots, s_{1,|S_1|}, s_{2,1}, \ldots, s_{2,|S_2|}\}$ which respects Definition~\ref{def:mcaisorder}:
Each vertex in $s_{1,1}, \ldots, s_{1,|S_1|}$ has no backward out-neighbor by the topological ordering of $S_1$; each $s_{2,1}, \ldots, s_{2,|S_2|}$ has no backward in-neighbor by the reverse topological ordering of $S_2$, and because vertices of $S_1$ can not have out-neighbors outside $S_1$; finally, every prefix of $s_{1,1}, \ldots, s_{1,|S_1|}, s_{2,1}, \ldots, s_{2,j}$ is connected, as $x, s_{2,1}, \ldots, s_{2,j}$ is connected, meaning that all vertices in $s_{2,1}, \ldots, s_{2,j}$ are connected to some vertex in $S_1$, that is itself connected.

We may now repeat this step by collapsing $S_1 \cup S_2$ into a vertex $x'$, and since $x'$ will be a source, take $S_3$ as all vertices reached by $x'$ in $S\setminus (S_1 \cup S_2)$, and take a topological order of $S_3 \cup \{x'\}$, which we append to the order obtained so far (excluding $x'$).

By iterating steps, we obtain an ordering $s_{1,1}, \ldots, s_{1,|S_1|}, s_{2,1}, \ldots, s_{2,|S_2|},$\linebreak $s_{3,1}, \ldots, s_{3,|S_3|} \ldots, s_{k,1}, \ldots, s_{k,|S_k|}$, with $k\le |S|$, that contains all vertices of $S$, and such that any prefix will induce a connected subgraph, and any $s_{i,j}$ will have no backward out-neighbors if $i$ is odd, and no backward in-neighbors if $i$ is even, thus there exist an ordering satisfying Definition~\ref{def:mcaisorder} (if a feasible order exists, a lexicographically minimum one must exist too).
\end{proof}

Finally, the proximity $\scap$ follows by Definition~\ref{def:proximity}.
We define the neighboring function as follows.

\begin{definition}[Neighboring Function for Maximal Connected Induced Acyclic Subgraphs]
For a solution $S$ and a vertex $v \in V(G)\setminus S$, we define
$$\neighs(S) = \bigcup\limits_{v\in V(G)\setminus S} \neighs(S,v)$$
Where
$\neighs(S,v) =\{ \comp( \ccomp_v(\{v\} \cup S\setminus N^+(v)) ), \comp( \ccomp_v(\{v\} \cup S\setminus N^-(v)) )\}$
\end{definition}

In other words, the function will add $v$ to $S$. $S\cup \{v\}$ is not acyclic, but all cycles must involve $v$, so we make it acyclic by removing either all the out-neighbors $N^+(v)$, which makes $v$ a target, or all its in-neighbors $N^-(v)$, which makes $v$ a source.
It then takes the connected component containing $v$ and feeds the result to $\comp(\cdot)$, to surely obtain a \mcais.


Consider now two solutions $S$ and $S^*$, and again let $\dv$ be the first vertex in the canonical order of $S^*$ which is not in $S\scap S^*$. More formally, let $S\scap S^* = \{s^*_1, \ldots, s^*_h\}$ and $\dv = s^*_{h+1}$. 

Let $S' = \comp( \ccomp_{\dv}(\{\dv\} \cup S\setminus N^+(\dv)) )$ and $S'' = \comp( \ccomp_{\dv}(\{\dv\} \cup S\setminus N^-(\dv)))$ be the two solutions generated by $\neighs(S,\dv)$.

By the canonical order of $S^*$, we have that $(S\scap S^*)\cup \{\dv\}$ is connected, and either $(S\scap S^*) \cap N^+(\dv) = \emptyset$ or $(S\scap S^*) \cap N^-(\dv) = \emptyset$.

It follows that if $(S\scap S^*) \cap N^+(\dv) = \emptyset$, then $(S\scap S^*)\cup \{\dv\} \subseteq \ccomp_{\dv}(\{\dv\} \cup S\setminus N^+(\dv)) \subseteq S'$, and otherwise we have $(S\scap S^*) \cap N^-(\dv) = \emptyset$, which means $(S\scap S^*)\cup \{\dv\} \subseteq \ccomp_{\dv}(\{\dv\} \cup S\setminus N^-(\dv))\subseteq S''$.

We thus have that either $|S'\scap S^*| > |S\scap S^*|$ or $|S''\scap S^*| > |S\scap S^*|$, which gives us the second necessary condition of proximity search.

Finally, it is straightforward to see that $\neighs(S)$ takes polynomial time, as its cost is bounded by $O(n)$ calls to $\comp(\cdot)$, which can be implemented in $O(mn)$, meaning that all conditions of Definition~\ref{def:searchable} are satisfied. Theorem~\ref{thm:mcais} follows.

\begin{theorem}\label{thm:mcais}
Maximal Connected Induced Directed Acyclic Subgraphs are proximity searchable, and can be listed $O(mn^2)$ time delay.
\end{theorem}

\subsection{Maximal Connected Edge-induced Directed Acyclic Subgraphs}
\newcommand{\mcaes}{\textsc{mcaes}\xspace}

We remark here that the structure can be adapted to the edge case, i.e., Maximal Connected Edge-induced Directed Acyclic Subgraphs (\mcaes). 

As the problem is still hereditary and acyclic subgraphs can be tested in linear time, we can implement the $\comp(\cdot)$ function in $O(m^2)$ time. The canonical order is as follows.

\begin{definition}[Canonical order for \mcaes]
Given a \mcaes $S$, let the canonical ordering of the \textit{vertices} of $G[S]$ according to Definition~\ref{def:mcaisorder} be $v_1, \ldots, v_{|V[S]|}$.

The canonical ordering of $S$ is obtained by selecting the edges of $S$ by increasing order with respect to their \textit{later} endpoint in the vertex order, and breaking ties by increasing order of the other (earlier) endpoint.
\end{definition}

We obtain a canonical ordering $e_1, \ldots e_{|S|}$ of $S$ with the following properties: take an edge $e_i = \{v_j,v_k\}$, assuming wlog $j<k$. All edges whose latter endpoint comes earlier than $v_k$ in the vertex order are preceding $e_i$ in the order, thus all edges in the induced subgraph $G[\{v_1, \ldots, v_{k-1}\}]$ will be in the prefix $e_1, \ldots e_{i}$ of the canonical ordering of $S$. By Definition~\ref{def:mcaisorder} $G[\{v_1, \ldots, v_{k-1}\}]$ is connected. Finally, the only other edges in $e_1, \ldots, e_{i}$ are those whose latter endpoint is $v_k$, so their earlier endpoint is in $\{v_1, \ldots, v_{k-1}\}$. Thus each prefix $e_1, \ldots e_{i}$ forms a connected (edge) subgraph, which is also acyclic as it is a subgraph of the acyclic subgraph $S$.

Furthermore, it also holds that, for the latter endpoint $v_k$ of $e_i$, either\linebreak $\{v_1, \ldots, v_{k-1}\} \cap N^+(v_k) = \emptyset$ or $\{v_1, \ldots, v_{k-1}\} \cap N^-(v_k) = \emptyset$. This implies that either $\{e_1, \ldots, e_{i-1}\} \cap N_E^+(v_k) = \emptyset$, or $\{e_1, \ldots, e_{i-1}\} \cap N_E^-(v_k) = \emptyset$, which gives us our neighboring function:

\begin{definition}[Neighboring Function for \mcaes]~

Let $S$ be a \mcaes and $e = (v_t, v_h)$ a directed edge in $E(G)\setminus S$ directed \textit{from} its tail $v_t$ \textit{to} its head $v_h$. Furthermore, let $N_E^+(v_h)$ and $N_E^-(v_t)$ be the out-edges and in-edges of $v_h$ and $v_t$, respectively. We define $ \neighs(S, v_t, v_h) = \{\comp(\ccomp_{v_t}(\{e\} \cup (S\setminus N_E^-(v_t) )) , \comp(\ccomp_{v_h}(\{e\} \cup (S\setminus N_E^+(v_h) ))\}$

And thus
$$\neighs(S) = \bigcup\limits_{e= (v_t, v_h) \in E(G)\setminus S}\neighs(S,v_t,v_h) $$
\end{definition}

In other words, we add $e$ to $S$, and try each of the two possibilities to obtain the latter vertex in the canonical order of $S^*$: if it is the tail $v_t$ of the edge, surely its backward out-neighborhood in the canonical order of $S^*$ is not empty as it contains $v_h$, so it's in-neighborhood must be, thus we can safely remove $N_E^-(v_t)$ to make $S\cup \{e\}$ acyclic. Conversely, if it is the head $v_h$ we can safely remove $N_E^+(v_h)$. We thus obtain $|S'\scap S^*| > |S\scap S^*|$ for some $S'\in \neighs(S)$.

We can observe that the cost $\ccost$ of a $\comp(X)$ call is $O(m^2)$ since we can test acyclicity in $O(m)$ time, which we do up to $m$ times, and finding and selecting the edges connected to $X$ take in total $O(m)$ time as well. As the neighboring function produces $O(m)$ solutions, we obtain:

\begin{theorem}\label{thm:mcaes}
Maximal Connected Edge-induced Directed Acyclic Subgraphs are proximity searchable, and can be listed $O(m^3)$ time delay.
\end{theorem}


\section{Proximity search in polynomial space}\label{sec:pspace-expl}\label{sec:pspace}
\newcommand{\good}{\ensuremath{\mathcal{F}}\xspace}
\newcommand{\univ}{\kern0.75pt\mathcal{U}\xspace}

Proximity search consists in a graph traversal, where the number of nodes corresponds to that of solutions. If we store the set of visited nodes, as done in the algorithms presented until now, it follows that the space requirement of the algorithm becomes exponential in $n$.


Techniques such as reverse-search are able to turn this graph into a rooted tree, that can be traversed without keeping track of visited nodes, by means of a parent-child relationship among solutions, thus achieving polynomial space. However, known instances of reverse search have de facto relied on the problem at hand being hereditary, and the input-restricted problem being solvable in polynomial time (respectively, polynomial total time) to obtain polynomial delay (polynomial total time). 
Recently, a generalization of reverse-search to non-hereditary properties has been proposed in~\cite{conte2019framework}: this allows us to induce a parent-child relationship for maximal solutions in any \textit{commutable} set system (a class of set systems which includes both hereditary and connected-hereditary properties), and obtain maximal listing algorithm with polynomial space, and whose delay is linked to the input-restricted problem.

In this section we show that, when suitable conditions are met, it is possible to get the best of both worlds: on one hand, using proximity search to overcome the burden of the input-restricted problem and achieve polynomial delay; on the other, using~\cite{conte2019framework} to induce a parent-child relationship among solutions and achieve polynomial space at the same time.

%
The final goal of the section is proving the following result.

\begin{theorem}
Let $(\univ,\good)$ be a commutable set system, and $\neighs(S,s)$ a canonical reconstruction function for a proximity search algorithm (see Definition~\ref{def:crecon}).
If the canonical order relative to the function $\neighs(S,s)$ satisfies the properties of a prefix-closed order (Definition~\ref{def:prefix-closed-order}), the maximal solutions of $(\univ,\good)$ can be enumerated without duplication in polynomial delay and polynomial space.
\end{theorem}

\subsection{Requirements and notation of~\cite{conte2019framework}}\label{sec:framework}


\newcommand{\seed}{\textsc{seed}\xspace}
\newcommand{\start}{\textsc{start}\xspace}
\newcommand{\parent}{\textsc{parent}\xspace}
\newcommand{\parind}{\textsc{pi}\xspace}
\newcommand{\core}{\textsc{core}\xspace}
\newcommand{\restr}{\textsc{restr}\xspace}

\newcommand{\gen}{\textsc{gen}\xspace}
\newcommand{\children}{\textsc{children}\xspace}
\newcommand{\cand}{\textsc{cand}\xspace}
\newcommand{\pcheck}{\textsc{p-check}\xspace}
\newcommand{\choice}{\textsc{choose}\xspace}
\newcommand{\erre}{\textsc{r}\xspace}

\let\oldnl\nl
\newcommand{\nonl}{\renewcommand{\nl}{\let\nl\oldnl}}


Let us briefly recall the requirements of~\cite{conte2019framework}. 
In a set system $(\univ,\good)$, $\univ$ is the 
\textit{ground set}, i.e., the elements constituting the solutions, and $\good$ defines the solutions, i.e., $S\in\good$ iff $S\subseteq \univ$ satisfies the property at hand. 

A set system is \textit{strongly accessible} if for any two distinct solutions $S,S'\in\good$ with $S \subset S'$, there exists an element $x\in S'\setminus S$ such that $S\cup \{x\}\in \good$. This is equivalent to saying that any non-maximal solution can be extended into a larger solution with a single element.

We say that a set system is \textit{commutable} if (i) it is strongly accessible, and (ii) it respects the \textit{commutable property}: for any $S,T\in \good$ with $S\subset T$, and any $a,b\in T\setminus S$, we have that $S\cup\{a\}\in\good \land S\cup\{b\}\in\good $ implies $S\cup\{a,b\}\in\good$. 
As mentioned in~\cite{conte2019framework}, it is straightforward to see that both hereditary and connected-hereditary properties correspond to commutable set systems.

Furthermore, we call $Z$ the set of ``singleton solutions'', i.e., $Z = \{e\in \univ : \{e\}\in\good\}$, and recall that in any strongly accessible set system $Z\cap S\ne \emptyset$ for any $S\in\good$. We also define, $S^+ = \{x : S\cup \{x\}\in\good \}$.

Given any commutable set system, we can obtain a maximal listing algorithm with two components. Firstly we need an efficient algorithm for solving the input-restricted problem. Secondly, to induce a parent-child structure we need what is called a \textit{family of prefix-closed orders} for the problem, satisfying the following properties:

\begin{definition}[Prefix-closed orders, from~\cite{conte2019framework}]
\label{def:prefix-closed-order}
Let $\Pi(X,v)$ be a family of orders parameterized by $X \in \good$ and $v \in X\cap Z$ such that $\Pi(X,v)$ yields a permutation of $X \cup X^+$. 
For $X \in \good$ and $v \in X\cap Z$, let us denote by $x^v_1,\ldots,x^v_k$ the elements of $X$ ordered according to $\Pi(X,v)$.\footnote{Note that $x_1=v$ and that $x_i, x_{i+1} \in X$ are not necessarily consecutive in $\Pi(X,v)$ as some elements from $X^+$ can be interleaved with them.} We call the family $\Pi$ \emph{prefix-closed} if for all $X \in \good$ and $v \in X \cap Z$, and $i \in \{1,\ldots,k-1\}$, the following properties hold:
\begin{description}
    \item[(first)] The minimal element is $v$, i.e., $x^v_1=v$.

    \item[(prefix)] The $i$-th prefix $X_i=\{x^v_1,\ldots,x^v_i\}$ of $X$ is a solution, i.e., $X_i \in \good$.

    \item[(greedy)] The element $x_{i+1}$ is the minimal element of $X_i^+\cap X$ with respect to the order $\Pi(X_i,v)$.
\end{description}
\end{definition}

As explained in~\cite{conte2019framework}, each subset $X$ of a maximal solution $S$ does not necessarily belong to $\good$ (as the set system is not necessarily hereditary). The \emph{first} property indicates that we can build $S$ starting from an element $v  \in S \cap Z$, whereas the \emph{greedy} property indicates that we can iteratively expand $X =\{v\}$ by considering the elements of $X \cup X^+$ in a prefix-closed order, so that at any point, the prefix $\{x_1,\ldots,x_j\}$ found so far is a solution thank to the \emph{prefix} property.
 
We use the shorthand notation $\prec^{t}_{S}$ to represent $\Pi(S,t)$, where $a\prec^{t}_{S} b$ for any two elements $a,b\in\univ$ means that $a$ occurs before $b$ in $\Pi(S,t)$.


Given a solution $S\in\good$ we define its \emph{seed}, $\seed(S)$, as the element of smallest id in $S\cap Z$, i.e., the element $s$ of smallest id in $S$ such that $\{s\}\in\good$. Observe that every non-empty solution $S$ of a strongly accessible set system has a seed: since $\emptyset \subset S$, there is some $s\in S \setminus \emptyset$ such that $\emptyset\cup \{s\} = \{s\}\in \good$.

The simplified notations $\prec_{S}$ corresponds to $\prec^t_{S}$ with $t = \seed(S)$. When $S$ is a maximal solution, $\prec_S$ defines an order $s_1,\ldots,s_{|S|}$ which is called the \emph{solution order} of $S$.

As in~\cite{conte2019framework}, we will also require a \textit{lexicographic} $\comp(\cdot)$ function: for a solution $S$, $\comp(S)$ must be obtained by iteratively adding to $S$ the smallest element in $S^+$ according to the order $\prec_S$ (i.e., the earliest in $\Pi(S,\seed(S))$), until $S^+$ is empty. The resulting solution is maximal by definition of strongly accessible set systems. We remark that this alternative definition of $\comp(S)$ still returns a maximal solution containing $S$, and is thus compatible with canonical reconstruction (Definition~\ref{def:crecon}).


Finally, given the canonical ordering $s_1, \ldots, s_{|S|}$ of $S$, the \textit{core} $\core(S)$ of $S$ is the \textit{longest} prefix $s_1, \ldots, s_i$ of this order such that $\comp(s_1, \ldots, s_i)\ne S$; its \textit{parent} is $\parent(S) = \comp(\core(S)) = \comp(s_1, \ldots, s_i)$; its parent index is $\parind(S) = s_{i+1}$, i.e., the element following the last one of the core. It follows by definition of parent that $\comp(\core(S)\cup \{\parind(S)\}) = \comp(s_1, \ldots, s_{i+1}) = S$.

The function $\parent(S)$ defines a forest among solutions, as every solution has a unique parent, except for the ones such that $\comp(\seed(S)) = S$ which are called \textit{roots}, and indeed correspond to the roots of the forest: these are linear in number (as each has a unique seed) and can be found by calling $\comp(\{u\})$ for any $u\in \univ$. The function $\children(P,w)$ lets us perform a traversal of this structure, since it will find all $S$ such that $P=\parent(S)$ and $w = \parind(S)$.


\subsection{Combining proximity search with \cite{conte2019framework}}\label{sec:pcombination}

\newcommand{\prefMacro}{\mathit{prefix}\xspace}
\begin{algorithm2e}[ht]
\caption{Polynomial-space proximity search}\label{alg:pspace}
\small%
\SetKwInOut{Input}{Input}
\SetKwInOut{Output}{Output}
%
\Input{Commutable set system $(\univ,\good)$\\Prefix-closed order family $\preceq^s_S$\\$\neighs(S,s)$ for canonical reconstruction based on $\preceq^s_S$}
\Output{All maximal $X\in\good$}
\BlankLine
\ForEach{$S$ \textnormal{such that} $\comp(\seed(S)) = S$\label{ln:ps:roots}}{
     $\rec( S )$
}

\SetKwProg{myproc}{Function}{}{}
  \nonl \myproc{$\rec(X)$}{
  \tcc{Output $X$ if depth is odd}
  
  \ForEach{$w\in \univ \setminus X$}{
    \ForEach{$S\in \children(X,w)$\label{ln:ps:rec}}{
    		$\rec(S)$
    }
  }
  \tcc{Output $X$ if depth is even}
}


\SetKwProg{myproc}{Function}{}{}
  \myproc{$\children (P, w)$}{
        \ForEach{$R\in \neighs(P,w)$\label{ln:ps:r}}{
            \ForEach{$s \in (R\cap Z) \setminus \{w\}$\label{ln:ps:s}}{
                $\prefMacro \gets \{x\in R : x\preceq_{R}^s w\}$\label{ln:ps:pref}\;
                $S \gets \comp(\prefMacro)$\;                
                \textbf{if} $\langle \parent(S), \parind(S), \seed(S), \erre(S)\rangle = \langle P,w,s,R\rangle$
                \textbf{then} \textbf{yield} $S$\label{ln:ps:pcheck}\;
            }
        }
    }

\myproc(\tcc*[f]{finds the first $R$ that can generate $S$}){$\erre (S)$}{
    $P\gets \parent(S)$\; 
    $w\gets \parind(S)$\;
    $s\gets \seed(S)$\;
    \ForEach{$R\in \neighs(P,w)$}{
        $\prefMacro \gets \{x\in R : x\preceq_{R}^s w\}$\;
        \lIf{$\comp(\prefMacro) = S$}{\textbf{return} $R$}  
    }
}

\end{algorithm2e}

When using proximity search in the \textit{canonical reconstruction} flavour, we use a canonical order to define the proximity by Definition~\ref{def:proximity}, and a suitable $\neighs(S,s)$ function such that together they satisfy Definition~\ref{def:searchable}.
In this section we show that we can combine proximity search and~\cite{conte2019framework} for commutable properties, if we can produce a canonical order for the canonical reconstruction that corresponds to the solution order induced by $\prec_S$.

We then show in Section~\ref{sec:pcord} that it is possible to meet these conditions for canonical orderings that are defined in a greedy way, e.g., by a BFS order like in bipartite subgraphs. 
%
Assuming that we meet these conditions, i.e., we have a $\neighs(S,s)$ function that fits canonical reconstruction (Definition~\ref{def:crecon}), based on a canonical order defined by a prefix-closed order $\preceq_S$, we define a variant of~\cite{conte2019framework}, showed in Algorithm~\ref{alg:pspace}.

The main idea behind this combination comes from the following observation: the parent $P = \parent(S) = \comp(\core(S))$ of $S$ is obtained from a prefix of $S$, and extending this prefix with $\parind(S)$, then applying $\comp(\cdot)$, gives us $\comp(\core(S)\cup\parind(S)) = S$ (see definitions in Section~\ref{sec:framework}). 
On the other hand, we will show that applying Definition~\ref{def:proximity}, $P \scap S$ is exactly $\core(S)$. 
Relying on the neighboring function $\neighs(P, \parind(S))$ of canonical reconstruction, and the \textit{core property} defined in~\cite{conte2019framework}, we are able to find the set $\core(S)\cup\parind(S)$, and finally obtain $S$.

We can now state:

\begin{theorem}
Given a commutable set system $(\univ,\good)$, a prefix-closed order family $\preceq^s_S$ for $(\univ,\good)$, and a function $\neighs(S,s)$ for canonical reconstruction (Definition~\ref{def:crecon}) based on $\preceq^s_S$, Algorithm~\ref{alg:pspace} enumerates all maximal solutions of $(\univ,\good)$ without duplication in polynomial delay.
\end{theorem}
\begin{proof}
To prove the correctness, we show that any $S$ is found in $\children(P,w)$ when $P=\parent(S)$ and $w = \parind(S)$.

We will first prove that there exists a solution $R \in \neighs(P, \parind(S))$ (on Line~\ref{ln:ps:r}) such that $\core(S)\cup\{\parind(S)\}\subseteq R$.

Consider the proximity $P\scap S$ by Definition~\ref{def:proximity}: the longer prefix of the solution order of $S$ that is completely in $P$ must include $\core(S)$ since $P=\comp(\core(S))$. If $w\in P$ then $\neighs(S,w)$ returns $P$ by Definition~\ref{def:crecon}, and indeed $P\supseteq \core(S)\cup\{w\}$.

Otherwise, $P$ does not include $\parind(S)$, meaning that $P\scap S = \core(S)$ and that $w$ is the canonical extender for $P,S$.
Using the neighboring function $\neighs(P, \parind(S))$ we obtain at least one solution $R \supseteq \core(S)\cup\parind(S)$.

Using the \textit{core property} defined in~\cite{conte2019framework}, we are able to use $R$ to retrieve $S$: It is proven that Lines~\ref{ln:ps:r}-\ref{ln:ps:pcheck} will find and output any solution $S$ such that $\core(S)\cup\{\parind(S)\}\subseteq R$, a condition which is guaranteed by what stated above.

The \textbf{if} on Line~\ref{ln:ps:pcheck} removes duplication: any $S$ is found only once out of all invocations of $\children(P,w)$: when $P = \parent(S)$, $w = \parind(S)$, $s=\seed(S)$, and $R = \erre(S)$. The function $\erre(S)$ simply aims at defining deterministically one single $R \supseteq \core(S)\cup\{\parind(S)\}$ once the other 3 variables have been fixed. It thus follows that this check is passed exactly once out of the whole execution of the algorithm for any solution (other than the roots, found on Line~\ref{ln:ps:roots}).

Line~\ref{ln:ps:roots} shows that, by definition, all the roots of the forest are explored by Algorithm~\ref{alg:pspace}. We just proved that Line~\ref{ln:ps:rec} discovers all the children of each visited node exactly once, which concludes the proof of the fact that Algorithm~\ref{alg:pspace} visits every maximal solution of $(\univ,\good)$ without duplication.
\end{proof}



It is also straightforward to see that each recursive call uses polynomial space, and no solution dictionary $\sol$ is maintained. However, the depth of the recursion tree is a factor in the space complexity too: to obtain a polynomial space guarantee, we further need to turn the recursive algorithm into a \textit{stateless} iterative one, as has been done in~\cite{conte2019framework}.

We can give a general bound with the following parameters: let $q$ be the maximum size of a solution; $\restrtime$ be the time required to solve $\neighs(P,w)$; $\restrbound$ a bound on the number of solutions returned by it; $\ccost$ be the time required to compute $\comp(X)$ and $\ocost$ the time required to compute the canonical order of $X\cup X^+$. As these bounds are all assumed to be polynomial, we observe their space requirements will be polynomial as well.

Thanks to the alternative output technique, the delay will be bounded by the cost of one iteration of $\rec(X)$, that is, $O(|\univ|)$ times the cost of $\children(P,w)$. In turn, the cost of $\children(P,w)$ is that of $\neighs(P,w)$, plus for each of the $O(\restrbound)$ solutions $R$ returned, the cost of processing Lines~\ref{ln:ps:s}-\ref{ln:ps:pcheck}. \cite{conte2019framework} proved that this can be done in  $O(q(\ocost+\ccost))$ time for the given definition of $\erre(S)$. However, our definition of $\erre(S)$ is different from the one in~\cite{conte2019framework}, and has a cost of $O(\restrtime + \restrbound\ocost)$ instead of $O(\ocost+\ccost)$. Thus, the total cost of processing Lines~\ref{ln:ps:s}-\ref{ln:ps:pcheck} is $O(q(\restrtime + \restrbound\ocost + \ccost))$.


We can thus claim the following:

\begin{theorem}\label{thm:prox-pspace-bound}
Given a commutable set system $(\univ,\good)$, a prefix-closed order family $\preceq^s_S$ for $(\univ,\good)$, and a function $\neighs(S,s)$ for canonical reconstruction (Definition~\ref{def:crecon}) based on $\preceq^s_S$, 
the maximal solutions of $(\univ,\good)$ can be enumerated in 
$O(|\univ|\restrtime + |\univ| \restrbound q(\restrtime + \restrbound\ocost + \ccost))$ time delay and polynomial space.
\end{theorem}

\subsection{BFS-based canonical reconstruction}\label{sec:pcord}\label{sec:bfsproximity}

In this section, we provide a technique to implement the result of Section~\ref{sec:pcombination} (Theorem~\ref{thm:prox-pspace-bound}), i.e., a canonical reconstruction order that matches the prefix-closed order requirements, and can be applied to hereditary and connected-hereditary properties. We call this technique \textit{BFS-based canonical reconstruction}.

While it is possibly not the only way to obtain a suitable order, it is worth defining formally as we will apply it to several problems in the following sections.

We will first define the order for connected-hereditary property, then exploit it to cover the hereditary case.\footnote{Notably, this implies that a BFS-based canonical reconstruction algorithm for the non-connected case immediately follows from one for the connected case.}

\begin{definition}[canonical-BFS order for connected-hereditary properties]\label{def:canonBFSconn}
Let $S$ be a solution of a connected-hereditary set system, and $v$ any element in $S$.
The canonical order $\Pi(S,v) = s_1,\ldots ,s_{|S\cup S^+|}$ is the lexicographical order of the tuples $\langle d_{v}(s_i),s_i\rangle$, where $d_{v}(s_i)$ is the distance between $s_i$ and $v$ in $G[S\cup \{s_i\}]$.
\end{definition}

In other words, we order nodes first by $d_{v}(s_i)$, i.e., their distance from $v$ in $G[S]$, and break ties by vertex id.
The same logic applies to nodes $x$ of $S^+$, for which we use the distance from $v$ in $G[S\cup\{x\}]$.
This defines $\preceq^s_S$. 

\begin{example*}
For the Maximal Connected Induced Bipartite Subgraph in Figure~\ref{fig:ex:bip} (b), the order $2,3,5,8,11,7,10$ (as defined in Section~\ref{sec:reconstruction}) is given by the tuples\linebreak $\langle0,2\rangle , \langle 1,3\rangle, \langle 1,5\rangle, \langle 2,8\rangle, \langle 2, 11\rangle, \langle 3,7\rangle, \langle 3,10\rangle$.
\end{example*}

We can observe how this canonical-BFS order $\Pi(S,v)$ satisfies the properties of Definition~\ref{def:prefix-closed-order}:\footnote{For completeness, we could equivalently observe that $d_{v}(s) = LAY^{v}_S(s)$ according to Definition~7 in~\cite{conte2019framework}.}

\begin{itemize}
    \item[\textbf{(first)}] The first element $s^v_1$ of $\Pi(S,v)$ is $v$, as $d_{v}(v)=0$ and $d_{v}(\cdot)\ge1$ for any other vertex.
    \item[\textbf{(prefix)}] Any prefix $S_i=\{s^v_1,\ldots,s^v_i\}$ of $\Pi(S,v)$ is connected (thus a solution), since for any $s_i$, the vertices on a shortest path in $G[S]$ to $s_1$ are at a smaller distance from $s_1$ and thus occur before $s_i$.
    \item[\textbf{(greedy)}] For any $z \in S_i^+ \cap S$, let $k$ be the distance between $v$ and $z$ in $G[S_i \cup \{z\}]$. Since $z\in S_i^+$, there must be some $w \in S_i\cap N(z)$ at distance $k-1$ from $v$ (in $G[S_i \cup \{z\}]$). This means that the distance between $v$ and $z$ in $G[S]$ is still $k$: otherwise, there would be a vertex $y\in S\setminus S_i$, i.e., after $w$ in the canonical order, that is a neighbor of $z$ and has distance $\le k-2$ from $v$ in $G[S]$; this leads to contradiction since $y$ would then need to come \textit{before} $w$ in the BFS order.  
    
\end{itemize}

It follows that the canonical-BFS order is a prefix-closed order.
now straightforward to see how this order satisfies the \textbf{(first)}, \textbf{(greedy)} and \textbf{(prefix)} properties of Definition~\ref{def:prefix-closed-order}, and essentially corresponds to the layer order defined in~\cite{conte2019framework}.

\begin{definition}[canonical-BFS order for hereditary properties]\label{def:canonBFSnon}
Let $S$ be a solution of a hereditary set system, and $v$ any element in $S$.
For each connected component $C_i$ of $G[S]$, we say the \textit{leader} of the component is $v$ if $C_i$ contains $v$, and otherwise the vertex of smallest id in $C_i$.

The canonical order $\Pi(S,v) = s_1,\ldots ,s_{|S\cup S^+|}$ (defined on $S\cup S^+$) is the lexicographical order of the tuples $\langle cid(S,s_i), d_l(S,s_i), s_i \rangle$, where for $s_i$ in the component $C_i$, $cid(S,s_i)$ is the id of the leader of $C_i$, or $0$ if this leader is $v$ (assuming wlog $0$ is smaller than any other id), and $d_l(S,s_i)$ is the distance from the leader of $C_i$ in $G[C_i]$. Observe how $s_1 = v$.
For a vertex $x$ in $S^+$, we use as $cid(S,x)$ and $d_l(S,x)$ the values obtained in $G[S\cup\{x\}]$.
\end{definition}

Less formally, we order each component by a BFS strategy as in the above case (since $G[C_i]$ is connected) using the leader as root (i.e., $s$ if the component contains $s$, or its smallest id vertex otherwise); then, we concatenate the sequences obtained by putting the one containing $s$ first, followed by the others ordered by id of their leader. 

\begin{example*}
For the Maximal Induced Bipartite Subgraph in Figure~\ref{fig:ex:bip} (d), the order is $1,2,7,8,11,10$, given by the tuples \linebreak $\langle 1,0,1\rangle , \langle 1,1,2\rangle, \langle 7,0,7\rangle, \langle 7,1,8\rangle, \langle 7,1,11\rangle, \langle 7,2,10\rangle$.
\end{example*}


Before proving that this defines a prefix-closed order, let us prove this auxiliary lemma:

\begin{lemma}
\label{lem:samecid}
Let $X$ be a solution and $X_i$ any prefix of its canonical order. The following facts hold:
\begin{itemize}
    \item $\forall z \in X_i^+\cap X, cid(X_i,z) \ge cid(X,z)$ 
    \item $cid(X_i ,x_{i+1}) = cid(X,x_{i+1})$
    \item $\forall z\in X_i^+\cap X, cid(X_i,z) = cid(X,z) \Rightarrow d_l(X_i,z) = d_l(X,z)$.
\end{itemize}
\end{lemma}
\begin{proof}
First, the leader of each connected component of $X_i$ is the same as the leader of the corresponding connected component of $X$ (since the leader is always the first element of the connected component in a solution order, and prefixes of components are connected as they are in a BFS order).

Moreover, an element $z$ in $X_i^+\cap X$ is either directly connected to a connected component of $X_i$, in which case it has the same leader in $X_i$ and in $X$ by what stated above, or it belongs to its own connected component in $G[X_i \cup \{z\}]$, in which case $z$ is its own leader in $G[X_i \cup \{z\}]$, meaning $cid(X_i, z) = z$. Since by definition $cid(X, z) \le z$, it follows that $cid(X, z) \le cid(X_i, z)$, proving the first statement.

We now prove that $cid(X_i, x_{i+1}) = cid(X, x_{i+1})$: Either $x_{i+1}$ is directly connected to the last connected component of $X_i$ (in which case we already proved the equality) or it isn't, in which case $cid(X_i, x_{i+1}) = x_{i+1}$. However, in this case $x_{i+1}$ must be its own leader by definition of the order, so it follows that $cid(X, x_{i+1}) = x_{i+1}$, proving the second statement.


Finally, consider $z\in X_i^+ \cap X$ such that $cid(X_i,z) = cid(X,z)$. 

If $cid(X,z) = z$ then $d_l(X_i,z) = d_l(X,z) = 0$; otherwise, let $x_l$ be the leader of $z$ in $X_i\cup\{z\}$: $z$ is in the same connected component $C_z$ as $x_l$ in $X$, and $X_i$ contains a prefix of the canonical-BFS order of $C_z$; by the properties of the canonical-BFS order, the shortest path from $x_l$ to $z$ is in this prefix, implying the third statement.
\end{proof}

We can now observe how this order for hereditary properties also satisfies the properties of Definition~\ref{def:prefix-closed-order}

\begin{itemize}
    \item[\textbf{(first)}] By definition $s_1$ is the first element.
    \item[\textbf{(prefix)}] As this order is defined for hereditary properties, it follows that any subset (hence every prefix) is also a solution.
    \item[\textbf{(greedy)}] We proved in Lemma~\ref{lem:samecid} that the tuple associated with each element of $X_i^+ \cap X$ with respect to $X_i$ is either the same or lexicographically greater than the tuple with respect to $X$. As the tuple for $x_{i+1}$ is the same, and since $x_{i+1}$ is the minimum of $X_i^+ \cap X$ with respect to the order in $X$, it follows that it's also the minimum of $X_i^+ \cap X$ with respect to the order in $X_i$.
\end{itemize}

We remark that it is possible to generalize this definition using different functions for $d(\cdot)$ and $d_l(\cdot)$, as long as monotone behaviour can be guaranteed, i.e., $d(X_i,x)$ (resp. $d_l(X_i,x)$) is less than or equal to $d(X,x)$ (resp. $d_l(X,x)$) when $X_i$ is a prefix of $X$.

\section{Polynomial space algorithms}\label{sec:pspace-algs}

In this section we apply the technique defined in Section~\ref{sec:pspace-expl}, and give polynomial-space-polynomial-delay proximity search algorithms, proving the bounds given in Theorem~\ref{thm:pspace}.

For the problems already solved in exponential space in the previous sections, we remark that it is simply necessary to define their canonical order as a canonical-BFS, then apply Theorem~\ref{thm:prox-pspace-bound}.

\subsection{Maximal Bipartite Subgraphs}

Looking at the canonical orders defined for Maximal Connected Induced Bipartite Subgraphs (Definition~\ref{def:bip-con-canon}) and Maximal Induced Bipartite Subgraphs (Definition~\ref{def:bip-non-canon}), we can see that their definitions match exactly those of canonical-BFS for connected-hereditary and hereditary properties (respectively, Definition~\ref{def:canonBFSconn} and Definition~\ref{def:canonBFSnon}).
We can thus immediately apply the polynomial space variant of the algorithm, and we proceed to compute its complexity.

The cost $\ocost$ for computing the canonical order will be $O(m)$ in all cases, as it corresponds to performing a BFS, while $\ccost$ corresponds to adding edges in a BFS order, which will take $O(m)$ on the connected version, but $O(m+n\iack(n))$ on the non-connected one due to the need to dynamically maintain the connected components. 
The neighboring function for both cases produces a constant number of neighboring solutions, meaning $\restrbound = O(1)$ and $\restrtime = O(\ccost)$.
At the same time, all operations require no more than $O(m)$ space. Applying Theorem~\ref{thm:prox-pspace-bound}, we obtain:

\begin{theorem}
Maximal Connected Induced Bipartite Subgraphs and Maximal Induced Bipartite Subgraphs of a graph $G$ can be enumerated via BFS-based canonical reconstruction (Algorithm~\ref{alg:pspace}) in $O(m)$ space and, respectively, $O(qnm) = O(n^2m)$ and $O(qn(m+n\iack(n))) = O(n^2(m+n\iack(n))) $ time delay.
\end{theorem}

\subsection{Maximal Induced Trees and Forests}\label{sec:indtrees}

As defined above, a forest is an acyclic undirected graph, and a connected forest is called a tree.
These are a special cases of $k$-degenerate subgraphs: $1$-degenerate subgraphs are precisely forests, and connected $1$-degenerate subgraphs are trees. However, it is worth consider these problems separately, since we can obtain algorithms with lower delay and polynomial space.

It should be observed that listing Maximal Induced Forests corresponds to listing minimal feedback vertex sets in undirected graph: if $S\subset V$ is a Maximal Induced Forest, $V\setminus S$ is a minimal feedback vertex set. A polynomial-delay solution for the enumeration of feedback vertex sets (and thus Maximal Induced Forests) has been proposed in~\cite{schwikowski2002enumerating}. This result, however, requires exponential space, and does not extend to Maximal Induced Trees.

Furthermore, while the algorithms proposed could be extended to enumerate maximal edge-induced trees and forests, we do not consider it: these correspond to just the spanning trees of a graph, which are already known to be enumerable in polynomial delay and even constant amortized time~\cite{shioura1997optimal}.

\paragraph{Canonical order and neighboring function}

Let $S$ be a maximal induced tree. 

We define its canonical as a canonical-BFS order (Definition~\ref{def:canonBFSconn}), i.e., the sequence $s_1, \ldots, s_{|S|}$ given by a BFS order of $G[S]$ rooted in the vertex $s_1$ of smallest id.


We then define the proximity by canonical reconstruction (Section~\ref{sec:reconstruction}),  and we can immediately observe that this order meets the requirements of  Section~\ref{sec:pspace}. Next, we focus on obtaining a suitable neighboring function.

\begin{definition}[Neighboring function for Maximal Induced Trees]~

We define $\neighs(S) = \bigcup\limits_{v\in V(G)\setminus S} \neighs(S,v)$.

Then, $\neighs(S,v)$ is defined as: 

\noindent$\neighs(S,v) = \{\comp(\ccomp_v(S\setminus N(v)\cup \{w,v\}) ) : w\in N(v) \cap S\}$
\end{definition}

The key property here is that each vertex $s_i\in S$ has a single neighbor preceding it in the canonical order, corresponding to its parent in the BFS.

Given two solutions $S,T$, let $t_1,\ldots,t_{|T|}$ be the canonical order of $T$, and $t_i$ be the canonical extender for $S,T$, i.e., the vertex for which $S\scap T = \{t_1,\ldots,t_{i-1}\}\subseteq S$ and $t_i\not\in S$. Furthermore, let $t_j$ be the parent of $t_i$ in the canonical BFS-order of $T$, observing that $t_j\in \{t_1,\ldots,t_{i-1}\}$.

To find a solution $S' \supseteq \{t_1,\ldots,t_{i}\}$, we can simply add $t_i$ to $S$, then remove all neighbors of $t_i$ \textit{except} $t_j$ so that we have again an acyclic subgraph, and finally discard every vertex not in the same connected component as $t_i$ (which will include $\{t_1,\ldots,t_{i}\}$). As we do not know which vertex is $t_j$, we of course try all $O(|N(t_i)|)$ possibilities, thus a suitable $S'$ is always found.

\paragraph{Complexity}

Firstly, we can use the $\neighs(S)$ function to build a proximity search algorithm whose delay is the cost of $\neighs(\cdot)$, and whose space is $O(\nsol \cdot n)$ (where $\nsol$ is the number of solutions). 

We show the cost of $\neighs(\cdot)$ -and the delay of the algorithm- to be $O(m^2)$ time: Observe that the cost $\ccost$ a $\comp(X)$ call is $O(m)$ time. We first compute the set of vertices adjacent to $X$, $P = \cup_{x\in X}N(x)$; for each vertex $v$, we simply need to check that it has exactly one neighbor in $X$, in $O(|N(v)|)$ time, and discard it otherwise. Whenever we add $v$ vertex to $X$, we add is neighbors to $P$ again in $O(|N(v)|)$ time. The total cost is $O(\sum_{v\in V(G)}|N(v)|) = O(m)$.

Now consider $\neighs(S,v)$: for each $w\in N(v)$, we must compute $\ccomp_v(S\setminus N(v)\cup \{w,v\})$, which takes $O(m)$, then apply $\comp(\cdot)$ which has the same complexity. The cost is thus $O(|N(v)|\cdot m)$. In turn, this means the cost of $\neighs(S)$ is $O(\sum_{v\in V(G)}|N(v)|\cdot m) = O(m^2)$.

Furthermore, as we are satisfying all conditions of Section~\ref{sec:pspace} (the order defined is a canonical BFS-order and the problem is connected-hereditary), we apply Theorem~\ref{thm:prox-pspace-bound} to obtain a BFS-based canonical reconstruction algorithm, with higher delay but polynomial space.

We observe that no component of the algorithm will require more than $O(m)$ space, and their time complexity is as follows: $\univ = O(n)$, $\restrtime = O(m\Delta)$, $\restrbound = O(\Delta)$ (but as observed above, $|\univ|\cdot \restrtime$ can be better bounded by $O(m^2)$, and $|\univ|\cdot \restrbound$ can be bounded by $O(m)$), $q= O(n)$, $\ocost = O(m)$ and $\ccost = O(m)$. 
The bound of Theorem~\ref{thm:prox-pspace-bound} thus resolves to 
$O(m^2 + m q(m\Delta + \Delta m + m)) = O(m q(m\Delta)) = O(m^2n^2)$ time.
We can thus conclude the following:

\begin{theorem}
The Maximal Induced Trees of a graph $G$ can be enumerated in $O(m^2)$-time delay using $O(\nsol n)$ space, or alternatively in $O(m^2n^2)$-time delay and $O(m)$ space.
\end{theorem}

\subsection{Maximal Induced Forests}

As showed in Section~\ref{sec:bfsproximity}, a BFS-based canonical reconstruction algorithm for the non-connected case immediately follows from the connected one.

For completeness, we show how the algorithm for Maximal Induced Forests is obtained:


The canonical order is obtained by Definition~\ref{def:canonBFSnon}, i.e., a canonical-BFS order of each connected component, where different components are then sorted by their vertex of smallest id.

The neighboring function is essentially obtained from the connected case by removing the use of the $\ccomp(\cdot)$ function (as we do not require solutions to be connected).

\begin{definition}[Neighboring function for Maximal Induced Forests]~

We define $\neighs(S) = \bigcup\limits_{v\in V(G)\setminus S} \neighs(S,v)$.

Then, $\neighs(S,v)$ is defined as: 

\noindent$\neighs(S,v) = \{\comp(S\setminus N(v)\cup w) : w\in N(v) \cap S\}$
\end{definition}

The complexity of the components of the algorithm is also inherently the same, with the only difference for the cost $\ccost$ of the $\comp(\cdot)$ function: when we add a vertex, we need to make sure that it does not have two neighbors in the same connected component, and update the connected components as we add vertices. The cost of $\comp(\cdot)$ will thus be $O(m + n\iack(n))$ time, obtained by the same logic as for Maximal Bipartite Subgraphs (see Section~\ref{sec:binon}), while the rest of the operations are exactly as in the connected case, thus bear the same cost.

We can conclude that $\neighs(S,v)$ takes $O(|N(v)|\cdot(m + n\iack(n))$ time, while $\neighs(S)$ takes $O(\sum_{v\in V(G)}|N(v)|\cdot (m + n\iack(n))) = O(m (m + n\iack(n)))$.

Again, we can obtain an exponential-space algorithm using canonical reconstruction proximity search whose delay is the cost of $\neighs(S)$, and a polynomial-space algorithm using BFS-based canonical reconstruction, whose delay is given by Theorem~\ref{thm:prox-pspace-bound}.

For the latter, the costs are obtained adapting the connected version with the new cost of $\comp(\cdot)$: $\univ = O(n)$, $\restrtime = O((m + n\iack(n))\Delta)$, $\restrbound = O(\Delta)$ (but $|\univ|\cdot \restrtime$ can be better bounded by $O(m(m + n\iack(n)))$, and $|\univ|\cdot \restrbound$ can be bounded by $O(m)$), $q= O(n)$, $\ocost = O(m)$ and $\ccost = O(m + n\iack(n))$. The bound of Theorem~\ref{thm:prox-pspace-bound} thus resolves to $O(m(m + n\iack(n)) +  m q ( (m + n\iack(n)) + \Delta m  + (m + n\iack(n)) ) ) = O( m q ( n\iack(n) + \Delta m ) )$ time, which we can again upper bound by $O(m^2n^2)$ time. We can thus conclude the following:

\begin{theorem}
The Maximal Induced Forests of a graph $G$ can be enumerated 
in $O(m^2n^2)$-time delay and $O(m)$ space.
\end{theorem}

\section{Conclusions}\label{sec:concl}

We presented proximity search, a technique for the design of efficient enumeration algorithms, based on defining and traversing a solution graph with bounded out-degree.
We showed several application cases, considering problems that did not allow efficient algorithms by known methods, and showing that these allow polynomial delay algorithms by proximity search.

We have provided a guideline, called \textit{canonical reconstruction}, aimed at factorizing the most effective ways to apply our technique, and facilitating the design of efficient algorithms.

We have further shown a technique that, under suitable conditions, allows us to design proximity search algorithms that require only polynomial space. The results are polynomial-delay and polynomial-space algorithms for several problems whose input-restricted problem cannot be solved in polynomial time, including non-hereditary ones.

This paper ``breaks the barrier'' of the input-restricted problem, showing that its complexity does not imply lower bounds in terms of time or space, nor even a trade-off between the two. This closes questions left open since~\cite{Cohen20081147}, furthering our understanding on the complexity of enumeration in set systems. 

At the same time, this reinvigorates the open question of which listing problems allow efficient algorithms and which do not, and to define a more complete theory of enumeration complexity.
On top of being a useful tool to design efficient algorithms for specific problems, we hope that this technique will be able to help us gain more insight into this general question.

\section*{acknowledgements} 
We wish to thank the anonymous reviewers for their thorough analysis of the paper, which helped us improve both its content and presentation. 
This work was partially supported by JST CREST, grant number JPMJCR1401, Japan and the Italian Ministry for Education and Research, under PRIN Project n. 20174LF3T8 AHeAD.

\bibliographystyle{plain}
\bibliography{arxiv}

\appendix
\section*{APPENDIX}

\newcommand{\mypar}[1]{\smallskip\noindent\textbf{#1}.~}
\section{Maintaining the solution set in proximity search}\label{sec:sol}

For completeness, we briefly describe how to efficiently maintain the $\sol$ set with well-known data structures.
In the following, let $\univ$ be the ground set (e.g., $V(G)$ for vertex-induced graph properties, or $E(G)$ for edge-induced graph properties).
Let $\nsol = |\sol|$ be the number of solutions in $\sol$, and let $s = \max_{S\in\sol}(|S|) \le |\univ|$ be the maximum size of a solution.
Recall $s\le n$ for vertex-induced graph properties, and $s\le m$ for edge-induced graph properties;

What we aim at showing is that the time for maintaining the solution set is negligible in all cases addressed in this paper: recall that any solution output by the neighboring function is maximalized, i.e., we apply a $\comp(\cdot)$ function which adds element to it until it is maximal. If we run $\comp(\emptyset)$, we can expect to add up to $s$ elements, so its worst-case complexity must be $\Omega(s)$ time.

\mypar{Binary Decision Diagram~\cite{knuth2011artBDD}} We can see it as a binary tree where leafs are all at depth $|\univ|$, and each root-to-leaf path defines a subset of $\univ$. We will have a space usage of $O(\nsol\cdot |\univ|)$, while the cost for addition or membership test of a solution will be $O(|\univ|)$ time.

This is sufficient for the purpose of our paper as we upper bound $s$ by $n$ (or $m$, for edge-induced subgraphs) in the complexity results, however it is possible to further improve this using a Trie:

\mypar{Trie~\cite{knuth2011artTrie}}
As above, a solution is represented by a root-to-leaf path. We only have nodes corresponding to including elements, so the depth will be $O(s)$, and so the space usage $O(\nsol \cdot s)$, however a node may have $O(|\univ|)$ children. If we keep these children sorted, we can look them up by binary search and have a cost for addition and membership of $O(s\log |\univ|)$; on the other hand, we can get constant time lookup using a hash table, and a cost for addition and membership of $O(s)$ time.


\end{document}